\documentclass{sig-alternate-10}

\usepackage[english]{babel}
\usepackage{amsmath}
\usepackage{amsfonts}
\usepackage{multirow}
\usepackage{url}
\usepackage{algorithm}
\usepackage[noend]{algpseudocode}
\usepackage{thmtools}
\usepackage{thm-restate}
\usepackage{etoolbox}
\usepackage{enumerate}

\usepackage[colorinlistoftodos]{todonotes}

\newtheorem{theorem}{Theorem}
\newtheorem{lemma}[theorem]{Lemma}
\newtheorem{proposition}[theorem]{Proposition}

\newtheorem{definition}[theorem]{Definition}

\newtheorem{remark}[theorem]{Remark}

\newtheorem{innercustomthm}{Theorem}
\newenvironment{customthm}[1]
  {\renewcommand\theinnercustomthm{#1}\innercustomthm}
  {\endinnercustomthm}

\newcommand{\calA}{\mathcal{A}}
\newcommand{\calD}{\mathcal{D}}
\newcommand{\G}{\mathcal{G}}
\newcommand{\N}{\mathbb{N}}
\newcommand{\bbP}{\mathbb{P}}
\newcommand{\calP}{\mathcal{P}}
\newcommand{\R}{\mathbb{R}}
\newcommand{\bbS}{\mathbb{S}}

\newcommand{\calS}{\mathcal{S}}
\newcommand{\Z}{\mathbb{Z}}
\newcommand{\disjind}{\textsf{DISJ+IND}}
\newcommand{\disj}{\textsf{DISJ}}

\newcommand{\poly}{\text{poly}}
\newcommand{\spoly}{\text{subpoly}}
\newcommand{\subpoly}{sub-polynomial}

\DeclareMathOperator{\argmin}{argmin}

\title{Streaming Space Complexity of Nearly All Functions of One Variable on Frequency Vectors}

\author{Vladimir Braverman\thanks{This material is based upon work supported in part by the National Science Foundation under Grant No. 1447639, by the Google Faculty Award and by DARPA grant N660001-1-2-4014. Its contents are solely the responsibility of the authors and do not represent the official view of DARPA or the Department of Defense.}\\ Johns Hopkins University\\ vova@cs.jhu.edu \and Stephen R.\ Chestnut\thanks{This material is based upon work supported in part by the
National Science Foundation under Grant No. 1447639.} \\ ETH Zurich\\ stephenc@ethz.ch
\and David P. Woodruff \\ IBM,
Almaden Research Center\\
dpwoodru@us.ibm.com
\and  Lin F.\ Yang\footnotemark[2] \\ Johns Hopkins University \\ lyang@pha.jhu.edu }

\date{}

\begin{document}
\maketitle

\begin{abstract}
A central problem in the theory of algorithms for data streams is to determine which functions on a stream can be approximated in sublinear, and especially sub-polynomial or poly-logarithmic, space. 
Given a function $g$, we study the space complexity of approximating
$\sum_{i=1}^n g(|f_i|)$, where $f\in\Z^n$ is the frequency vector of a turnstile stream.
This is a generalization of the well-known frequency moments problem, and previous results apply only when $g$ is monotonic or has a special functional form. 
Our contribution is to give a condition such that, except for a narrow class of functions $g$, there is a space-efficient approximation algorithm for the sum if and only if $g$ satisfies the condition.
The functions $g$ that we are able to characterize include all convex, concave, monotonic, polynomial, and trigonometric functions, among many others, and
is the first such characterization for non-monotonic functions. 
{Thus, for nearly all functions of one variable, we answer the open question from the celebrated paper of Alon, Matias and Szegedy (1996).}

\end{abstract}

\section{Introduction}\label{sec:intro}

One of the main open problems in the theory of data stream computation is to characterize functions of a frequency vector that can be computed, or approximated, efficiently.
Here we characterize nearly all functions of the form $\sum_{i=1}^n g(|v_i|)$, where $v=(v_1,\ldots,v_n)\in\Z^n$ is the frequency vector of the stream. This is a generalization of the famous frequency moments problem, where $g(x)=x^k$ for some $k\geq 0$, described by Alon, Matias, and Szegedy~\cite{alon1996space}, who also asked:\\
{\em ``It would be interesting to determine or estimate
the space complexity of the approximation of $\sum_{i=1}^nv_i^k$
for non-integral values of $k$ for $k < 2$, or the
space complexity of estimating other functions of the numbers $v_i$.''}
The first question was answered by Indyk and Woodruff~\cite{indyk2005optimal} and Indyk~\cite{indyk2000stable} who were the first to determine, up to polylogarithmic factors, the space complexity of the frequency moments for all $k>0$.  
We address the second question.

Braverman and Ostrovsky~\cite{braverman2010zero} and Braverman and Chestnut~\cite{braverman2015universal}, answer the second question for sums of the form given above when $g$ is monotone.
We extend their characterizations to nearly all nonmonotone functions of one variable.
Specifically, we characterize the set of functions $g$ for which there exists a sub-polynomial $(1\pm\epsilon)$-approximation algorithm for the sum above.
Our results can be adapted to characterize the set of functions approximible in poly-logarithmic, rather than sub-polynomial, space.
Among the main techniques we use is the layering method developed by Indyk and Woodruff~\cite{indyk2005optimal} for approximating the frequency moments, and one may view our results as an exploration of the power of this technique.

Our results also partially answer an open question of Guha and Indyk at the IIT Kanpur workshop in 2006~\cite{sublinearNumberFive} about characterizing sketchable distances - our results characterize nearly all distances that can be expressed in the form $d(u,v)=\sum_i g(|u_i-v_i|)$. 
Besides the aforementioned work on the frequency moments and monotone functions, other past work on this question was done by Guha, Indyk, and McGregor~\cite{guha2007sketching} as well as Andoni, Krauthgamer, and Razenshteyn~\cite{andoni2014sketching}. 
Both of those papers address the same problem, but in domains different from ours.
In the realm of general streaming algorithms, a result of Li, Nguyen, and Woodruff~\cite{li2014turnstile} shows that for any function $g$ with a $(1\pm\epsilon)$-approximation algorithm implementable in sub-polynomial space, there exists a linear sketch that demonstrates its tractability. However, a drawback of that work is that
the space complexity of maintaining the sketching matrix, together with computing the output, may be
polynomially large. 
Our paper deals with a smaller class of functions than \cite{li2014turnstile}, but we provide a zero-one law and explicit algorithms.

{While several of our lower bounds can be shown via reduction from the standard index and set-disjointness communication problems, to prove lower bounds for the widest class of functions we develop a new class of communication problems and prove lower bounds for them, which may be of independent interest. The problem, which we call {\sf ShortLinearCombination}, is: given a finite set $\{0, a_1, \ldots, a_r, b\}$ of possible frequencies, is there
a frequency in the stream of value $b$? Perhaps surprisingly, the communication 
complexity of {\sf ShortLinearCombination} depends on the
magnitudes of the coefficients in any linear combination expressing $b$ in terms of $a_1, \ldots, a_r$. We develop
optimal communication bounds for this problem. This hints at the subtleties in proving lower bounds for 
a wide class of functions $g$, since by defining $g(b) \geq n \cdot \max_i g(a_i)$, the communication problem
just described is a special case of characterizing the streaming complexity of all functions $g$. 
Our lower bounds provide a significant generalization of the set-disjointness problem, which has
found many applications to data streams.}

\subsection{Applications}

We now describe potential applications of this work.

\subsubsection{Log-likelihood Approximation}

One use for the approximation of non-monotonic functions is the computation of log-likelihoods.
Here, the coordinates of the streamed vector~$v$ are taken to be i.i.d.\ samples from a discrete probability distribution.
The likelihood of the vector $v$, under a distribution with probability mass function $p(x)$, $x\in\Z$, is $L(v) = \prod_{i=1}^n p(v_i)$, and its log-likelihood is $\ell(v) = -\sum_{i=1}^n \log p(v_i)$.
Notice that $\ell(v)$ is of the form $\sum_{i=1}^n g(v_i)$, for $g(x)=-\log p(x)$.
When $p(\cdot;\theta)$ is the distribution of a non-negative random variable or if $p(x;\theta)$ is symmetric about $x=0$, our results can be applied to determine whether there is a streaming algorithm that efficiently approximates $\ell$.
In general, $-\log p(x)$ is not a monotonic function of $x$.
For example, $p(x)=\lambda\frac{x^{\alpha}}{x!}e^{-\alpha} + (1-\lambda)\frac{x^{\beta}}{x!}e^{-\beta}$ is a mixture of two Poissons, which is generally not monotonic. 
For any constants $\lambda, \alpha,\beta>0$, the Poisson mixture log-likelihood $-\log p(x)$ satisfies our three criteria (to be described later), hence the log-likelihood of the stream can be approximated in poly-logarithmic space.
Continuous distributions can be handled similarly by discretization.

Suppose that the distribution $p$ comes from a family of probability distributions $p(\cdot;\theta)$ parameterized by $\theta\in\Theta$.
When an efficient approximation algorithm exists, we describe a linear sketch from which a $(1\pm\epsilon)$-approximation $\hat{\ell}$ to $\ell(v)$ can be extracted with probability at least $2/3$.
The form of the sketch is independent of the function $g$, so if $-\log p(x;\theta)$ satisfies our conditions for approximability for all $\theta\in\Theta$, we derive an approximation to each of the log-likelihoods $\ell(v;\theta)$ that are separately correct with probability $2/3$ each.
If $\Theta$ is a discrete set then this allows us to find an approximate maximum likelihood estimator for $\theta$ with only an $O(\log |\Theta|)$ factor increase in storage.
Recall, the maximum likelihood estimator is $\hat{\theta} = \argmin_{\theta\in\Theta} \ell(\theta,v)$.
Indeed, as usual we can drive the error probability down to any $\delta>0$ with $O(\log 1/\delta)$ repetitions of the algorithm, so an additional factor of $O(\log|\Theta|)$ in the space complexity is enough.
If $\hat{\ell}(\theta;v)$ denote our approximations, then $\hat{t} = \argmin_{\theta\in\Theta}\hat{\ell}(\theta;v)$ is the approximate maximum likelihood estimate, which has the guarantee that $\ell(\hat{t};v)\leq (1+\epsilon)\ell(\hat{\theta};v)$.
When $\Theta$ is not finite one may still be able to apply our results by discretizing $\Theta$ with $\poly(n)$ points and proceeding as above; whether this works depends on the behavior of $p(x;\theta)$ as a function of $\theta$.

\subsubsection{Utility Aggregates}

Consider an online advertising service that charges its customers based on the number of clicks on the ad by web users.
More clicks should result in a higher fee, but an exceptionally high number of clicks from one user make it likely that he is a bot or spammer. 
Customers may demand that the service discount the cost of these spam clicks, which means that the fee is a non-monotonic function of the number of clicks from each user.

The ads application is an example of a non-monotonic utility function.
Many utility functions are naturally non-monotonic, and they can appear in streaming settings.
For example, extremely low or high trading volume in a stock indicates abnormal market conditions.
A low volume of network traffic might be a sign of damaged or malfunctioning equipment whereas a high volume of network traffic is one indicator of a denial-of-service attack.

\subsubsection{Database Query Optimization}

A prominent application for streaming algorithms, indeed one of the initial motivations~\cite{alon1996space}, is to database query optimization.
A query optimizer is a piece of software that uses heuristics and estimates query costs in order to intelligently organize the processing of a database query with the goal of reducing resource loads (time, space, network traffic, etc.).
Scalable query optimizers are needed for scalable database processing~\cite{wu2011query,jahani2011automatic}, and poor query efficiency has been cited as one of the main drawbacks of existing implementations of the MapReduce framework~\cite{lee2012parallel}.
Streaming algorithms and sketches are a natural choice for such heuristics because of their low computational overhead and amenability to distributed processing.
By expanding the character of statistics that can be computed efficiently over data streams, our results could allow database query optimizers more expressive power when creating heuristics and estimating query costs.

\subsubsection{Encoding Higher Order Functions}

An obvious generalization of the problem we consider is to replace the frequency vector $f_{i}$, $i\in[n]$, with a frequency matrix $f_{ij}$, $i\in[n]$ and $j\in[k]$, and ask whether there is a space efficient streaming algorithm that approximates $\sum_{i=1}^n g(f_{i,1},\ldots,f_{i,k})$. 
These could, for example, allow one to express more complicated database queries which first filter records based on one attribute and then sum up the function values applied to another attribute on the remaining records.

Let us demonstrate that when the coordinates of $f$ are bounded and $k$ is not too large, we can replace this sum with a function of a single variable.
Suppose that $0\leq f_{ij}\leq b-1\in\N$, for all $i,j$.
Upon receiving an update to coordinate $(i,j)$ we replace it with $b^j$ copies of $i$.
The new frequencies are polynomially bounded if $b^k=\poly(n)$.

Call the new frequency vector $f'\in\Z^n$.
The sequence of values $f_{i1},\ldots,f_{ik}$ is immediately available as the base-$b$ expansion of the number $f'_i$, so we are able to write $g'(f'_i) = g(f_{i1},\ldots,f_{ik})$, where $g'$ first recovers the base-$b$ expansion and then applies $g$.
Our desire is to approximate $\sum_i g'(f'_i)$.
Given even a well behaved function $g$, it is very unlikely that $g'$ will be monotone, hence the need for an approximation algorithm for non-monotonic functions if this approach is taken.
It is also very likely that, because of the construction, $g'$ has high local variability.
This is a significant problem for algorithms that use only one pass over the stream (as we show by way of a lower bound), but we also present a two pass algorithm that is not sensitive to local variability.

\subsection{Problem Definition}
A stream of length $m$ with domain $[n]$ is a list $D = \langle(i_1, \delta_{1}), ({i_2}, \delta_{2}), \ldots (i_m, \delta_{m})\rangle$, where $i_j \in [n]$ and $\delta_j\in\Z$ for every $j\in [m]$. 
The \emph{frequency vector} of a stream $D$ is the vector $V(D)\in\Z^n$ with $i^{th}$ coordinate $v_i=\sum_{j:i_j=i}\delta_j$.
A processor can read the stream some number $p\geq 1$ times, in the order in which it is given, and is asked to compute a function of the stream.
We allow the processor to use randomness and we only require that its output is correct with probability at least $2/3$.

Our algorithms are designed for the \emph{turnstile} streaming model.
In particular, there exists $M\in\N$ and it is promised that $V(D)\in\{-M,-M+1,\ldots,M\}^n$ and the same holds for any prefix of the list $D$.
Our lower bounds, on the other hand, hold in the more restrictive \emph{insertion-only} model which has $\delta_j=1$, for all $j$.
We use $\calD(n, m)$ to denote the set of all turnstile streams with domain~$[n]$ and length at most~$m$. 

Given a function $g: \mathbb{R}\rightarrow \mathbb{R}$ and vector $V=(v_1,v_2,\ldots,v_n)$, let  \[g(V):=\sum_{i=1}^ng(|v_i|).\]

\begin{definition}
Given $g:\Z_{\geq0}\to\R$, $\epsilon>0$, and $D\in\calD(n,m)$ the problem of $(g,\epsilon)$-SUM on stream $D$ is to output an estimate $\hat{G}$ of $g(V(D))$ such that 
\[P\left((1-\epsilon)g(V(D))\leq \hat{G}\leq(1+\epsilon)g(V(D))\right)\geq 2/3.\]
We often omit $\epsilon$ and refer simply to $g$-SUM when the value of $\epsilon$ is clear from the context.
\end{definition}
The choice of $2/3$ here is arbitrary, since given a $g$-SUM algorithm the success probability can be improved to $1-1/\poly(n)$ by repeating it $O(\log{n})$ times in parallel and taking the median outcome.

Our goal is to classify the space complexity of $g$-SUM. 
We ask: for which functions~$g$ can $(g,\epsilon)$-SUM be solved in space that is sub-polynomial in $n$?
We call such functions {\em tractable}.
An $\Omega(\epsilon^{-1})$ lower bound applies for nearly every function $g$, thus we only require that the algorithm solve $(g,\epsilon)$-SUM whenever $\epsilon$ decreases sub-polynomially.
Our main results separately answer this question in the case in which the processor is allowed one pass over the stream and in the case in which $O(1)$ passes are allowed. 
We are able to classify the space complexity of almost all functions, though a small set of functions remains poorly understood.

\subsection{Organization}

Section~\ref{sec: our results} gives a high-level description of our results and states our main theorems.
We put off some of the more technical definitions until Section~\ref{sec: preliminaries}, which gives the definitions needed to precisely state our main results as well as background on the main techniques for our proofs.
Next, Section~\ref{sec: weak zero-one} includes the intuition and proofs for our main theorems.  
It describes one and two-pass algorithms and lower bounds.
Section~\ref{sec: exotic} and Appendix~\ref{app:nearlyperiodic} discuss the functions that we are unable to characterize in more depth.

\section{Our Results}\label{sec: our results}
To begin with, we separate the class of functions $g:\Z_{\geq0}\to\R$ into two complementary classes: ``normal functions'' and ``nearly periodic functions''.
In fact, we do this differently for $1$-pass $g$-SUM than for multi-pass $g$-SUM. 
We define $\bbS$-normal functions, for studying one pass streaming space complexity, and $\bbP$-normal functions for multiple passes.
The complements of these two sets are the $\bbS$-nearly periodic functions and $\bbP$-nearly periodic functions, respectively.
The distinction between $\bbS$ and $\bbP$ is not important to understand our results at a high level, so we omit them for the rest of this section.
See Section~\ref{sec: exotic} and Appendix~\ref{app:nearlyperiodic} for more about this class of functions.
We postpone the technical definitions until Section~\ref{sec: preliminaries}, and give an informal overview here.

We are able to characterize the normal functions based on three properties which we call \emph{slow-jumping}, \emph{slow-dropping}, and \emph{predictable}.
Slow-jumping depends on the rate of increase.
Roughly, a function is slow-jumping if it doesn't grow much faster than $y=x^2$ at every scale.
Slow-dropping governs the rate of decrease of a function.
A function is slow dropping if decreases no faster than sub-polynomially at every point.
Finally, if a function is predictable then it satisfies a particular tradeoff between its growth rate and local variability. Interestingly, not just the properties themselves are important for our proofs, but so is the interplay between them.
Precise definitions for these properties are given in Section~\ref{sec: preliminaries}.

Our proofs are by finding heavy hitters, for the upper bound, and by reductions from communication complexity for the lower bound.
See Section~\ref{sec: techniques} for more thorough discussion of how the three properties relate to heavy-hitters and the other techniques we use.

\subsection*{Zero-One Laws for Normal Functions}

A nonnegative function $f$ is called a \subpoly{} function if, for all ${\alpha>0}$, $\lim_{x\to\infty}x^\alpha f(x)=\infty$ and $\lim_{x\to\infty}x^{-\alpha}f(x)=0$.
Formally, we study the class of functions $g$ such that $g$-SUM can be solved with \subpoly{} accuracy $\epsilon=\epsilon(n)$ using only an amount of 
space which is sub-polynomial.
We call such functions $p$-pass tractable if the algorithm uses $p$ passes. 
In this paper we prove the following theorems.
\begin{theorem}[1-pass Zero-One Law]
\label{thm:norm1passalg}
A function $g\in\G$ is 1-pass tractable and normal if and only if it is slow-jumping, slow-dropping, and predictable.
\end{theorem}
\begin{theorem}[2-pass Zero-One Law]
\label{thm:norm2passalg}
A function $g\in\G$ is 2-pass tractable and normal if and only if it is slow-dropping and slow-jumping.
\end{theorem}
The main difference between the two theorems is that predictability is not needed in two passes.
The message is that large local variability can rule out 1-pass approximation algorithms but not 2-pass approximation algorithms.
Theorem~\ref{thm:norm2passalg} extends to any subpolynomial number of passes.

Most functions one encounters are normal, which include convex, concave, monotonic, polynomial, or trigonometric functions, as well as functions of regular
variation and unbounded Lipschitz-continuous functions.

\section{Preliminaries}\label{sec: preliminaries}
In describing the requirements and space efficiency of our algorithms for $g$-SUM, we use the set of sub-polynomial functions.
\begin{definition}

A function $f:\mathbb{R}_{\geq0}\rightarrow\mathbb{R}_{\geq0}$ is {\em \subpoly{}} if for any $\alpha>0$, the following two conditions are satisfied: a) $\lim_{x\rightarrow\infty}x^\alpha f(x) = \infty$ and b) $\lim_{x\rightarrow\infty}x^{-\alpha} f(x) = 0$.
\end{definition}

We use $\spoly(x)$ to represent the set of \subpoly{} functions in the variable $x$. 
Examples of sub-polynomial functions include polylogarithmic functions and some with faster growth, like $2^{\sqrt{\log{n}}}$.
The set of poly-logarithmic functions is a subset of the \subpoly{}~functions.
Formally, we study the class of functions $g$ such that $(g,\epsilon)$-SUM can be solved with any accuracy $\epsilon \in\spoly(n)$, using only \subpoly{}~space.
Our algorithms assume an oracle for computing $g$ and that the storage required for the value $g(x)$ is sub-polynomial in $x$.
Regardless, our sketches are sub-polynomial in size, but without these assumptions it could take more than sub-polynomial space to compute the approximation from the sketch.
We restrict ourselves to the case $M\in\poly(n)$. 
This restriction is reasonable because if it grows, say, exponentially in $n$, then an algorithm can store the entire frequency vector and compute $g(V(D))$ exactly in polylog$(nM)$ space. 

\begin{definition}
A function $g:\Z_{\geq 0}\rightarrow \mathbb{R}$ is {\em $p$-pass tractable} if given any sub-polynomial function $h(x)$ and any $\epsilon\ge 1/h(nM)$ there exists a sub-polynomial function 
$h^*(x)$ and a $p$-pass algorithm $\mathcal{A}$ that solves $(g,\epsilon)$-SUM for all streams in $\calD(n,m)$ and $n,M\geq 1$ using no more than $h^*(nM)$ space in the worst case.
\end{definition}

Replacing the storage requirement, which is $h^*(nM)$, with $h^*(n)\log M$ would also be natural, but since we consider $M\in\poly(n)$ the two definitions are equivalent.
It simplifies our notation a little bit to write $h^*(nM)$.
Unless otherwise noted, for the rest of this paper we require that $g(0) = 0$ and $g(x) > 0$, for all $x>0$. 
The choice $g(0)=0$ is equivalent to requiring that $g(V(D))$ does not depend on the particular choice of $n$ in the model; specifically it avoids the following behavior: if $g(0)\neq 0$ then given a single stream $D$ the value of $g(V(D))$ differs depending on the choice of the dimension, even though the stream remains the same.
Functions with $g(0)\neq 0$ may be of interest for some applications. The laws for these functions are very similar to the case when $g(0)=0$ and we provide them in Appendix~\ref{app:caseg00}. 

As shown in \cite{ChestnutThesis}, functions with $g(x)=0$, for some integer $x>0$, are not $n^\alpha$-pass tractable for any $\alpha<1$, unless $g$ is nonnegative and periodic with period $\min\{x>0\mid g(x)=0\}$, which must exist if $g$ is periodic and $g(0)=0$.
It is also shown in \cite{ChestnutThesis} that if a non-linear function $g$ takes both positive and negative values, then $g$ is not $n^\alpha$-pass tractable, for any $\alpha<1$. 
Without loss of generality, we can assume $g(1) = 1$ because a multiplicative approximation algorithm for the function $g^\prime(x) := g(x)/g(1)$ is also a multiplicative approximation algorithm for $g$.  
Finally, for simplicity of notation we will often extend the domain of $g$ symmetrically to $\Z$, i.e., setting $g(x)=g(-x)$, for all $x\in\Z_{\geq0}$, which allows us the simpler notation $g(|v_i|)=g(v_i)$.
In summary, we study the functions in the class
\[\mathcal{G} \equiv \{g:\Z_{\geq0}\rightarrow \mathbb{R},  g(0) = 0,  g(1) = 1, \forall x>0, g(x) >0\}\]

Our results are based on three characterizations of the variation of the functions. 
We now state the technical definitions mentioned in Section~\ref{sec: our results}.
\begin{definition}
A function $g\in\G$ is \emph{slow-jumping} if for any $\alpha > 0$, there exists $N >0$,  such that for any $x<y$ if $y\ge N$, then $g(y)\le   \lfloor\frac{y}{x}\rfloor^{2+\alpha}x^\alpha g(x)$.
\end{definition}
Slow-jumping is a characterization of the rate of growth of a function. 
Examples of slow-jumping functions are $x^{p}$, for $p\leq 2$, $x^22^{\sqrt{\log x}}$, and $\left(2+\sin x\right)x^2$, while functions like $2^x$ and $x^p$, for any $p>2$, are not slow jumping because they grow too quickly. 

\begin{definition}
A function $g\in\G$ is {\em slow-dropping} if for any $\alpha >0$ there exists $N>0$, such that for any $x<y$, if $y \ge N$, then $g(y)\ge g(x)/y^\alpha$.
\end{definition}
Whether a function is slow-dropping is determined by its rate of decrease.
The functions\footnote{$\mathbf{1}(\cdot)$ denotes the indicator function.} $(\log_2 1+x)^{-1}\mathbf{1}(x>0)$ and $\left(2+\sin x\right)x^2$ are slow-dropping, but any function with polynomial decay, i.e. $x^{-p}$, for $p>0$, is not slow-dropping.

Given a function $g$, $x\in\N$, and $\epsilon>0$ define the set:
$\delta_\epsilon(g, x):= \{y\in\N \mid |g(y)-g(x)|\leq \epsilon g(x)\}.$ 
\begin{definition}
A function $g$ is \emph{predictable} if for every $0<\gamma<1$ and sub-polynomial $\epsilon$, there exists $N$ such that for all $x\geq N$ and $y\in [1,x^{1-\gamma})$ with $x+y\notin\delta_{\epsilon(x)}(g,x)$ we have $g(y)\geq x^{-\gamma}g(x)$.
\end{definition}
Predictability is governed by the local variability of a function and its rate of growth.
For example, the function $g(x)=x^2$ is predictable because $g(x+y)/g(x) = (1+\frac{y}{x})^2\approx 1$ when $y\ll x$, or more precisely $\pi_\epsilon(x)\supseteq \{y\mid |x-y|\leq \epsilon x/3\}$. 
That function has low local variability.
On the other hand, the function $(2+\sin x)\mathbf{1}(x>0)$ is locally highly variable but still predictable.
Notice that even $x+1\notin \pi_\epsilon(x)$, for say $\epsilon=0.01$, but $g(1)/g(x) \geq 1/3$ so the inequality in the definition of predictability is satisfied.
A negative example is the function $(2+\sin x)x^2$, which is not predictable because it varies quickly (by a multiplicative factor of 3) and grows with $x$.

As we will show, the three conditions above can be used to characterize nearly every function in $\G$ in both the single-pass and $O(1)$-pass settings.
We remark here that the characterization has already been completed for monotonic functions.
The tractability condition for nondecreasing $g$ proved by~\cite{braverman2010zero} is equivalent to $g$ being slow-jumping and predictable.
For nonincreasing functions, it is a consequence of~\cite{braverman2015universal} that polynomially decreasing functions are not tractable while sub-polynomially decreasing functions are tractable.
The ``nearly periodic'' functions are formally defined next. 

Let $\bbS$ be the set of non-increasing sub-polynomial functions on domain $\Z_{\geq0}$ and $\bbP$ be the set of strictly increasing polynomial functions on $\Z_{\geq0}$.
We now define the sets of $\bbS$-nearly periodic functions and $\bbP$-nearly periodic functions.
The motivation for these definitions jumps ahead to our lower bounds, but we will explain it here.
The reductions used in some of our lower bounds boil down to finding three integers $x<y$ and $x+y$ such that $g(x)\gg g(y)$ and $g(x)\not\approx g(x+y)$.
The reductions fail when $g(x)\approx g(x+y)$, where the meaning of ``$\approx$'' depends on whether we are bounding $1$-pass or multi-pass algorithms.
For the 1-pass reduction to fail it is necessary that $\frac{1}{g(x)}|g(x)-g(x+y)|$ decreases as $x$ increases, hence the $\bbS$-nearly periodic functions.
The 2-pass reduction fails as long as $\max(\frac{g(x)}{g(x+y)},\frac{g(x+y)}{g(x)})$ is not polynomially large in $y$, hence the $\bbP$-nearly periodic functions.

\begin{definition}

Given a set of functions $\calS$, call
$g(x)$ {\em $\calS$-nearly periodic}, if the following two conditions are satisfied.
\begin{enumerate}
\item There exists $\alpha >0$ such that for any constant $N>0$ there exists $x, y\in \mathbb{N}$, $x<y$ and $y\ge N$ such that  $g(y) \le g(x) / y^\alpha$. Call such a $y$ an {\em $\alpha$-period of $g$};
\item For any $\alpha > 0$ and any error function $h\in \calS$ there exists $N_1>0$ such that for all $\alpha$-periods $y\ge N_1$ and all $x<y$ such that $g(y) y^\alpha \le g(x)$, we have $|g(x+y)-g(x)|\le \min\{g(x), g(x+y)\}h(y).$
\end{enumerate}
A function $g$ is {\em $\calS$-normal} if it is not $\calS$-nearly periodic.  
\end{definition}
The first condition states that the function is not slow-dropping.
For example, if $g$ is bounded then there is an increasing subsequence along which $g$ converges to 0 polynomially fast.
To understand the second condition, take $h$ to be a decreasing function or a small constant.
In loose terms, it states that if $x<y$ and $g(x)\gg g(y)$, then $g(x)\approx g(x+y)$.
The choice of the set of functions $\calS$ determines the relative accuracy implied by ``$\approx$''.
We will apply the definition with $\calS$ set to either $\bbS$ or $\bbP$, that is, with either subpolynomially or polynomially small relative accuracy.

An $\bbS$-nearly periodic function (it is also $\bbP$-nearly periodic) can be constructed as follows.
For each $x\in\N$, let $i_x = \max\{j: 2^j | x\}$ and let $g_{np}(0)=0$ and $g_{np}(x) = 2^{-i_x}$. 
For example, $g_{np}(1)=1$, $g_{np}(2)=1/2$, $g_{np}(3)=1$, $g_{np}(4)=1/4$, etc.
The proof that $g$ is nearly periodic appears in Appendix~\ref{example:nearly periodic}, but for now notice that
it satisfies the first part of the definition because $g_{np}(2^k)=2^{-k}$ and, as an example of the second part, we have $g_{np}(2^k+1) = g_{np}(1)=1$.
One may guess that a streaming algorithm for such an erratic function requires a lot of storage.
In fact, $g_{np}$ can be approximated in polylogarithmic space!
Appendix~\ref{example:nearly periodic} has the proof.

The following containment follows directly from the definition of the nearly periodic functions.
\begin{proposition}\label{prop:pnorm in snorm}
Every $\bbS$-nearly periodic function is also $\bbP$-nearly periodic.
Every $\bbP$-normal function is $\bbS$-normal.
\end{proposition}

\subsection{Heavy Hitters, Communication Complexity, and CountSketch}

Our sub-polynomial space algorithm is based on the Recursive Sketch of Braverman and Ostrovsky~\cite{braverman2013generalizing}, which reduces the problem to that of finding heavy hitters.

\begin{definition}
Given a stream $D\in\calD(n,m)$ and a function $g$, call $j\in[n]$ a \emph{$(g,\lambda)$-heavy hitter} of $V(D)$ if $g(|v_j|)\ge \lambda\sum_{i\neq j} g(|v_i|)$.
We will use the terminology $\lambda$-heavy hitter when $g$ is clear from the context.
\end{definition}

Given a heaviness parameter $\lambda$ and approximation accuracy $\epsilon$, a heavy hitters algorithm returns a
set $\{i_1,i_2,\ldots\}$ containing every $(g,\lambda)$-heavy hitter and also returns a $(1\pm\epsilon)$-approximations to $g(v_{i_j})$ for each $i_j$ in the set. 
Such a collection of pairs is called a $(g,\lambda,\epsilon)$-cover.

\begin{definition}
Given a stream  $D\in\calD(n,m)$, a $(g,\lambda, \epsilon)$-cover of $V(D)$ is a set of pairs $\{(i_1, w_1), (i_2, w_2),$ $\ldots, (i_t, w_t)\}$, where $t\le n$ and the following are true:
\begin{enumerate}
\item $\forall k\in[t]$, $|w_{k} - g(|v_{i_k}|)| \le \epsilon g(|v_{i_k}|)$, and 
\item if $j$ is a $(g,\lambda)$-heavy hitter of $V(D)$, then there exists $k\in[t]$ such that $i_k = j$.
\end{enumerate}
\end{definition}

Formally, a $p$-pass $(g,\lambda,\epsilon,\delta)$-heavy hitter algorithm is a $p$-pass streaming algorithm that outputs a $(g,\lambda, \epsilon)$-cover of $V(D)$ with probability at least $(1-\delta)$.

\begin{theorem}[\cite{braverman2013generalizing}]\label{thm: recursive sketches}
Let $\lambda=\frac{\epsilon^2}{\log^3 n}$ and $\delta = {\frac{1}{\log{n}}}$.
If there exists a $(g,\lambda,\epsilon,\delta)$-heavy hitter algorithm using $s$ bits of space, then
there exists a $(g,\epsilon)$-SUM algorithm using $O(s\log n)$ bits of space.
\end{theorem}

A similar approach has been used in other streaming algorithms; it was pioneered by Indyk and Woodruff~\cite{indyk2005optimal} for the problem of approximating the frequency moments.
Their algorithm uses the streaming algorithm CountSketch of Charikar, Chen, and Farach-Colton~\cite{charikar2002finding}, and ours will as well.
Given $\lambda$, $\epsilon$, and $\delta$, a CountSketch data structure using space $O(\frac{1}{\lambda\epsilon^2}\log{\frac{n}{\delta}}\log{M})$ gives an approximation~$\hat{v}_i$ to the frequency of every item~$i\in[n]$ with the following guarantee: with probability at least $(1-\delta)$, $|v_i-\hat{v}_i|\leq \epsilon\sqrt{\lambda F_2}$ for all $i\in[n]$.
By adjusting the parameters slightly we can treat the output of CountSketch$(\lambda,\epsilon,\delta)$ as a set of $k=O(1/\lambda)$ pairs $(i_j,\hat{v}_{i_j})$ such that $\{i_j\}_{j=1}^{k}$ contains all of the $\lambda$-heavy hitters for $F_2$ and $|v_{i_j}-\hat{v}_{i_j}|
\leq \epsilon\left(\sum_{j>k}\bar{v}_j^2\right)^{1/2}
\leq \epsilon\sqrt{\lambda F_2},$
for all $j=1,2,\ldots,k$, where $(\bar{v}_1,\ldots,\bar{v}_n)$ is a reordering of the frequency vector from largest to smallest absolute magnitude.

We establish our lower bounds using communication and information complexity. 
Several of our lower bounds can be shown by reduction from standard INDEX and DISJ problems, as well as a 
combination of them (related combined problems were used in ~\cite{ChestnutThesis,li2013tight}). 
In the INDEX$(n)$ problem, there are two players Alice, given a set $A\subseteq[n]$, and Bob, given $b\in[n]$. 
Alice sends a message to Bob, and with probability at least $2/3$, Bob must determine whether $b\in A$.
Any randomized one-way protocol for INDEX$(n)$ requires Alice to send $\Omega(n)$ bits \cite{knr99}. 
In an instance of $\disj(n,t)$ there are $t$ players each receiving a subset of $[n]$ with the promise that either the sets are pairwise disjoint, or there is a single item that every set contains and 
other than this single item, the sets are pairwise disjoint.
The players must determine which case they are in.
The randomized, unrestricted communication complexity of $\disj(n,t)$ 
is $\Omega(n/t)$~\cite{bar2002information, chakrabarti2003near, gronemeier2009asymptotically, j09}.
In the $\disjind(n,t)$ problem, 
$t+1$ players are given sets $A_1,A_2,\ldots,A_{t+1}\subseteq[n]$, such that $|A_{t+1}|=1$, with the promise that either the sets are disjoint or there is a single element contained in every set and they are otherwise disjoint.
A standard argument shows the one-way communication complexity of $\disjind(n,t)$ is $\Omega(n/t(\log n))$ (see Theorem \ref{thm: communication complexity of disjind}). 

{ 
To obtain stronger lower bounds for the nearly periodic function sums, we use the information complexity
framework.
\begin{definition}[{\sf ShortLinearCombination} problem] \label{def:SLC}
Let $u=(u_1, u_2, \ldots u_r)$ be a vector in $\mathbb{Z}^r$ for some integer $r$.
Let $d>0$ be an integer not in  $\{|u_1|, |u_2|, \ldots, |u_r|\}$. A stream $S$ with frequency vector $v$ is given to the player, $v$ is promised to be from $V_0=\{u_1, u_2, \ldots, u_r, 0\}^n$ or $V_1=\{\text{EMB}(v, i, e)\mid v\in V_0, i\in[n], e\in\{-d, d\}\}$, where EMB$(v,i,e)$ is to embed $e$ onto the $i$-th coordinate. 
The $(u,d)$-DIST problem is to distinguish whether $v\in V_0$ or $v\in V_1$. 
\end{definition}
This problem is in fact a generalization of the DISJ$(n,t)$ problem. 
We show a lower bound for this problem of $\Omega(n/q^2)$ where $q=\sum_{i=1}^r{|q_i|}$ and $q_i$ are the integers satisfying $\min_{q}\{q_1, q_2,\ldots, q_r\mid \sum_{i=1}^rq_i u_i = d\}$. 
For more details, we refer the reader to Appendix~\ref{app:communication}. 
}

\section{Zero-One Laws for Normal Functions}\label{sec: weak zero-one}
First, we develop some intuition about why the three conditions, slow-dropping, slow-jumping, and predictable, characterize tractable normal functions.
Next comes the two pass and one pass algorithms followed by the lower bounds.
Finally, we prove Theorem~\ref{thm:norm1passalg} and~\ref{thm:norm2passalg} in Section~\ref{sec:zoproofs}.

\subsection{Our Techniques}\label{sec: techniques}

The upper and lower bounds are both established using heavy hitters.
Using a $\frac{\epsilon^2}{\log^3 n}$-heavy hitters subroutine that identifies each heavy hitter $H$ and also approximates $g(f_H)$, the algorithm of Braverman and Ostrovsky~\cite{braverman2013generalizing} (Theorem~\ref{thm: recursive sketches}) solves $g$-SUM with $O(\log{n})$ storage overhead.
We will show that if $g$ is tractable then a subpolynomially sized CountSketch suffices to find heavy hitters for $g$.
On the lower bounds side, typical streaming lower bounds, going back to the work of Alon, Matias, and Szegedy, are derived from reductions to communication complexity that only require the streaming algorithm to detect the presence of a heavy hitter.
Thus the problem of characterizing tractable functions can be reduced to characterizing the set of functions admitting subpolynomial heavy hitters algorithms.

Now let us explain how each of the three properties relates to heavy hitters.
For this discussion let $\alpha=\epsilon^2/\log^3n$ be the heaviness parameter needed by Braverman's and Ostrovsky's algorithm and let $H\in[n]$ be the identity of some $\alpha$-heavy hitter.

\paragraph{Slow-jumping and Slow-dropping}
A function that satisfies these two conditions cannot increase much faster than a quadratic function nor can it decreasy rapidly.
This means that if $H$ is a heavy hitter for $g$ then it is also a heavy hitter for $F_2$.
Concretely, suppose that there exists an increasing subpolynomial function $h$ such that $g(y)/g(x)\leq (y/x)^2 h(y)$ and $g(x)\geq g(y)/h(y)$, for all $x\leq y$, so that $g$ is slow-jumping and slow-dropping (Propositions~\ref{prop: slow dropping to sub-polynomial} and~\ref{prop: slow jumping to sub-polynomial} show that such an $h$ always exists if $g$ is slow-jumping and slow-dropping).
A bit of algebra shows that $g(f_H)\geq \alpha\sum_{i} g(f_i)$ implies $f_H^2\geq \frac{\alpha}{h(M)^2} \sum_{j}f_j^2$, so $H$ is a $\alpha':=\frac{\alpha}{h(M)^2}$-heavy hitter for $F_2$.
CountSketch suffices to identify all $\alpha'$-heavy hitters for $F_2$ with $O(\alpha'\log^2{n})$ bits, so this gives an $\alpha$-heavy hitters algorithm for $g$ that still uses subpolynomial space.

If a function is not slow-jumping then a heavy hitter may be too small to detect and we prove a lower bound by reduction from $\disjind$ in a very similar fashion to the original lower bound for approximating the frequency moments by Alon, Matias, and Szegedy.

If a function is not slow-dropping then $f_H$ may be hidden below many small frequencies.
In this case, a reduction from the communication complexity of $\textsf{INDEX}$ usually works to establish a one-pass lower bound (two player $\disj$ can be used for a multipass lower bound).
Here is a preview of the reduction that will also explain the genesis of the nearly periodic functions.
Suppose $g(1)=1$ and there is an increasing sequence $n_k\in\N$ such that $g(n_k)\leq 1/n_k$.
Alice and Bob can solve an instance $(A\subseteq[n_k],b\in[n_k])$ of $\textsf{INDEX}$ as follows.
Alice creates a stream where every item in $A$ has frequency $n_k$ and the items in $[n_k]\setminus A$ do not appear.
To Alice's stream Bob adds one copy of his item and they jointly run a streaming algorithm to approximate $g(f)$.
The result is an approximation to either $|A|g(n_k)+1$ or $(|A|-1)g(n_k) + g(n_k+1)$.
If $g(n_k+1)$ differs significantly from $g(n_k)+1\approx 1$ then Bob can determine whether $b\in A$ by checking the approximation to $g(f)$.
The reduction fails if $g(n_k+1)\approx 1$ and such a fuction may be tractable!
A prime example is $g_{np}$ from Section~\ref{sec: preliminaries}, which we demonstrate in Appendix~\ref{example:nearly periodic}.\\

At this point, one might have in mind our two pass heavy hitters algorithm.
The algorithm is to use a CountSketch on the first pass to identify a subpolynomial size set of items containing all $\alpha$-heavy hitters and then use the second pass just to tabulate the frequency of each of those items exactly.
The details and proof of correctness for that algorithm are in Section~\ref{sec:two pass}.
Local variability of the function is irrelevant for two passes because we can compute the frequencies exactly during the second pass, and that is why a function need not satisfy predictability to be two pass tractable.
A one pass algorithm will have to get by with approximate frequencies.

\paragraph{Predictable}
Predictability is a condition on the local variability of the function. 
It roughly says that if $y \ll x$ then either $g(x+y)\approx g(x)$ or $g(y)$ is on the same scale\footnote{More precisely, $g(y)$ is not much smaller than $g(x)$, if $g$ is slow-dropping then $g(y)$ cannot be much larger than $g(x)$, either.} as $g(x)$.
The first conclusion has an obvious interpretation; it means that substituting an approximation to $x$ yields an approximation to $g(x)$.  
CountSketch returns an approximation to the frequency of every heavy hitter, so, in the first case, we can just as well use the approximation and there is no need to determine the frequency exactly.

The second conclusion allows local variability, but it must be accompanied by some ``global'' property of $g$, namely, $g$ does not change a whole lot over the interval $[y,x]$.
For an illustration of how predictability works, suppose that $H$ is a heavy hitter and there is some $y\ll f_H$ such that $g(y) \geq 2g(f_H)$.
Our algorithm cannot substitute $f_H+y$ for $f_H$ because this could lead to too much error.
But, approximating $f_H$ to better accuracy than $\pm y$ using a CountSketch, one would generally need more than $(f_H/y)^2$ counters, which could be very large.
Notice that $y$, if it occurs as a frequency in the stream and suppose it does, is a heavy hitter, and if $g$ is slow-jumping and slow-dropping then $y$ is also a heavy hitter for $F_2$ by our previous argument.
It follows that the error in the frequency estimate derived from the CountSketch is less than $y$, in particular an estimate of $f_H$ derived from the CountSketch will be more accurate than $\pm y$.
The result is that we get an error estimate for the CountSketch that is improved over a na\"ive analysis. 
The approximate frequency $\hat{f}_H$ satisfies $g(\hat{f}_H)\approx g(f_H)$ and  is accurate enough for Braverman and Ostrovsky's algorithm.
When $g$ is locally variable but not predictable, we get a one pass lower bound by a reduction from $\textsf{INDEX}$.\\

\subsection{Two Pass Algorithm}\label{sec:two pass}
We we will now show that any $\lambda$-heavy hitter for a slow-dropping, slow-jumping, $\bbS$-normal function is also $F_2$ $\lambda'$-heavy for $\lambda'$ modestly smaller than $\lambda$.
This means that we can use the CountSketch to identify heavy hitters in one pass for the Recursive Sketch.
For a heavy hitters algorithm we also need a $(1\pm\epsilon)$-approximation to each item's contribution to $g$-SUM, we can easily accomplish this in the second pass by exactly tabulating the frequency of each heavy hitter identified in the first pass.
This algorithm works as long as we identify every heavy hitter in the first pass, which CountSketch guarantees, and the space required remains small provided that we do not misclassify too many non-heavy hitters (so that we do not tabulate too many values in the second pass).
The procedure is formally defined in Algorithm~\ref{algo: 2-pass}. Note that by Proposition~\ref{prop:pnorm in snorm} $\bbP$-normal functions are also $\bbS$-normal, hence, the algorithm works for slow-dropping and slow-jumping $\bbP$-normal functions.

The next two propositions describe equivalent definitions of slow-dropping and slow-jumping that are useful in describing the algorithm.
\begin{proposition}\label{prop: slow dropping to sub-polynomial}
$g\in\G$ is slow-dropping if and only if there exists a sub-polynomial function $h$ such that for all $y\in \N$ and $x<y$ we have $g(x)\leq g(y)h(y)$.
\end{proposition}
\begin{proof}
The ``if'' direction follows immediately from the definition of the class of sub-polynomial functions.

For the ``only if'', suppose that $g(x)$ is slow-dropping. 
Specifically, for any $\alpha >0$ there exists $N>0$ such that for all $y\geq N$ and $x<y$ we have $y^\alpha g(y)\ge g(x)$. 
Let $N_i$ be the least integer such that $y^{1/i}g(y)\ge g(x)$, for all $y>N_i$ and $x<y$. 
The sequence $N_i$ is nondecreasing.
Consider the increasing function $i(y) = \max \{i: N_i\le y\} $ and set $h(y) = y^{1/{i(y)}}$. 
For all $x<y$ we have $h(y)g(y)\geq g(x)$ by construction.

To see that $h$ is sub-polynomial, let $\beta>0$.
First, $h\geq1$ hence $h(y)y^{\beta}\to\infty$.
Second, let $i^*\geq 2/\beta$ then for all $y\geq N_{i^*}$ we have 
\[h(y)\leq y^{1/i^*}\leq y^{\beta/2}.\]
Thus $h(y)y^{-\beta}\to 0$, which completes the proof.
\end{proof}

\begin{proposition}\label{prop: slow jumping to sub-polynomial}
$g\in\G$ is slow-jumping if and only if there exists a sub-polynomial function $h(x)$ such that for any $x<y$ we have $g(y)\le {\lfloor y/x\rfloor}^{2}h(\lfloor y/x\rfloor x)g(x)$.
\end{proposition}
\begin{proof}
Again, the ``if'' direction follows immediately from the definition of the class of sub-polynomial function.
The reverse direction follows from the same argument as Proposition~\ref{prop: slow dropping to sub-polynomial}.
\end{proof}

Given $g$ and a sub-polynomial accuracy $\epsilon$, Propositions~\ref{prop: slow dropping to sub-polynomial} and \ref{prop: slow jumping to sub-polynomial} each imply the existence of a non-decreasing sub-polynomial function.
By taking the point-wise maximum, we can assume that these are the same sub-polynomial function $H:\N\to\R$.
In particular $H$ satisfies:
\begin{itemize}
\item $g(y)\geq g(x)/H(y)$, for all $x<y$, and
\item $g(y)\leq (y/x)^2y^\alpha H(y)g(x)$, for all $x<y$. \end{itemize}
The space used by our algorithm depends on the sub-polynomial functions governing the slow-dropping and slow-jumping of $g$.  
If these are polylogarithmic then the algorithm is limited to  polylogarithmic space.

\begin{lemma}\label{lem: g heavy implies F2 heavy}
Let $g$ be a function that is slow-jumping and slow-dropping.
There exists a sub-polynomial function $h$ such that for any $D\in\calD(n,m)$ with frequencies $v_1,v_2,\ldots,v_n$, if $g(v_i)\geq \lambda \sum_j g(v_j)$ then \[v_i^2\geq \frac{\lambda}{h(|v_i|)}\sum_{|v_j|<|v_i|} v_j^2.\]
\end{lemma}
\begin{proof}
By Proposition~\ref{prop: slow jumping to sub-polynomial}, there exists a nondecreasing sub-polynomial function $h$ such that for all $y\in\N$ and $x<y$ we have $g(y)\leq (y/x)^2h(y) g(x)$.

For any $j$ that satisfies $|v_j|<|v_i|$, we have $g(v_i)\leq g(v_j)(\frac{v_i}{v_j})^2 h(|v_i|)$.
Therefore,
\begin{equation*}
g(v_i) 
\ge \lambda \sum_{j}g(v_j)
\ge \lambda \sum_{|v_j|<|v_i|}\frac{g(v_i)}{h(|v_i|)}\left(\frac{v_j}{v_i}\right)^{2}.
\end{equation*}
By rearranging, we find 
\begin{equation*}
v_i^2\ge \frac{\lambda}{h(|v_i|)} \sum_{|v_j|<|v_i|}v_j^2\ge \frac{\lambda}{h(|v_i|)} \sum_{|v_j|<|v_i|}v_j^2.
\end{equation*}
\end{proof}

\begin{algorithm}
\begin{algorithmic}[1]
\Procedure{2-Pass Heavy Hitters}{$g$, $\lambda$, $\epsilon$, $\delta$}
\State {\bf First Pass:} 
\State \hspace{\algorithmicindent} $S\gets$CountSketch$(\frac{\lambda}{2H(M)}, \frac{1}{3},\delta)$ discarding the frequency estimates
\State {\bf Second Pass:} 
\State \hspace{\algorithmicindent} 
   Compute $v_j$ for all $i\in S$
\State \Return $(j,v_j)$, for all $j\in S$ 
\EndProcedure
\end{algorithmic}
\caption{A 2-pass $(g,\lambda,0,\delta)$-heavy hitters algorithm.}
\label{algo: 2-pass}
\end{algorithm}

Algorithm~\ref{algo: 2-pass} computes CountSketch\footnote{A CountSketch is a matrix with $r$ and $b$ rows.  See~\cite{charikar2002finding} for details about these parameters.} with $r=O(\log{n})$ and $b = O(\frac{H(M)}{\lambda\epsilon^2})$
over the stream and extracts the $2H(M)/\lambda$ items with the highest estimated frequencies.
Denote this list as $(i_j,\hat{v}_{i_j})_{j=1}^{2H(M)/\lambda}$.
The 2-pass algorithm then discards the estimated frequencies $\hat{v}_{i_j}$ and tabulates the true frequency of each item $i_j$ on the second pass.

\begin{lemma}\label{lem: identifying heavy hitters}
If $g$ is a function that is slow-dropping and slow-jumping then the CountSketch used by Algorithm~\ref{algo: 2-pass} finds all of the $(g,\lambda)$-heavy hitters.
\end{lemma}
\begin{proof}
From the slow-dropping condition and Lemma~\ref{lem: g heavy implies F2 heavy}, for every $x<y\in \N$ we have
\begin{enumerate}
\item\label{it: decreasing part} $g(y)H(y)\geq g(x)$ and 
\item\label{it: increasing part} $g(v_i)\geq \lambda \sum_j g(v_j)$ implies \[v_i^2\geq \frac{\lambda }{H(M)}\sum_{|v_j|<|v_i|} v_j^2.\]
\end{enumerate}

Let $D\in\calD(n,m)$, suppose $i$ satisfies $g(v_i)\ge \lambda \sum_{j}^ng(v_j)$.
Since
\[g(v_i)\geq \lambda \sum_{|v_j|\geq |v_i|} g(v_j)\geq \frac{\lambda}{H(M)} \sum_{|v_j|\geq |v_i|} g(v_i),\]
there are at most $H(M)\lambda^{-1}$ items with frequencies as large or larger than $v_i$ in magnitude.

Thus, a CountSketch for $\lambda/2H(M)$-heavy hitters, as is used by Algorithm~\ref{algo: 2-pass}, serves to identify the $\lambda$-heavy heavy hitters for $g$ in one pass.
\end{proof}

\begin{theorem}
\label{thm:snorm2passalg}
A function $g\in\G$ is 2-pass tractable and $\bbS$-normal if it is slow-dropping and slow-jumping.
\end{theorem}
\begin{proof}
If a function is slow-dropping then it is $\bbS$-normal.
By Theorem~\ref{thm: recursive sketches}, it is sufficient to show that Algorithm~\ref{algo: 2-pass} is a sub-polynomial memory $(g,\lambda,\epsilon,\delta)$-heavy hitters algorithm for $\lambda=\frac{\epsilon^2}{\log^3 n}$, $\epsilon =\frac{1}{H(Mn)}$, and $\delta\le \frac{1}{\log n}$.

This is a $2$-pass $(g,\lambda,0,\delta)$-heavy hitters algorithm, and it requires  $O(\frac{h(M)}{\epsilon^2}\log^4{n}\log{nM}\log\log{n})$ space.
The algorithm is correct with probability at least $(1-\delta)$ because this is true for the CountSketch.
Therefore, $g$ is a 2-pass tractable function.
\end{proof}

\subsection{One Pass Algorithm}

The only impediment to reducing the two pass algorithm to a single pass is local variability of the function. 
The algorithm must simultaneously identify heavy hitters and estimate their frequencies, but local variability could require tighter estimates than seem to be available from a sub-polynomial space CountSketch data structure.
If an $\bbS$-normal function is predictable then we almost immediately overcome this barrier. 
In particular, the definition means that given any point $x$ and a small error $y=o(x)$ either $|g(x+y)-g(x)|\leq \epsilon g(x)$ or $g(y)$ is reasonably large, and while it is clear how the first case helps us with the approximation algorithm, the second is less clear.
The slow-dropping and slow-jumping conditions play a role, as we now explain.

\begin{proposition}\label{prop: predictable to sub-polynomial}
A function $g$ is predictable if and only if for every sub-polynomial $\epsilon>0$ there exists a sub-polynomial function $h$ such that for all $x\in\N$ and $y\in[1,x/h(x))$ satisfying $x+y\notin\delta_{\epsilon(x)}(g,x)$ it holds that $g(y)\geq g(x)/h(x)$.
\end{proposition}
\begin{proof}
The ``if'' direction follows directly from the definition of sub-polynomial functions.

To prove the ``only if'' direction, let $\gamma_i= 1/i$. Since $g$ is predictable, there exists $N_i$ such that either $x+y\in\delta_{\epsilon(x)}(g,x)$ or $g(y)\geq x^{-\gamma}g(x)$ for $x\geq N_i$ and $y\in[1,x^{1-\gamma})$.
Let $i(z) = \max\{i\in\N \mid N_i\leq z\}$, for $z\geq N_1$, and define $h(z) = z^{1/i(z)}$, for $z\geq N_1$, and $h(z) = \max_{1\leq x',y'<N_1}g(x')/g(y')$, for $z<N_1$.
Then $h$ is a sub-polynomial function.
The claim is that $h$ satisfies the conclusion of the lemma.
Let $x\in\N$ and $y\in[1,x/h(x))$. 
If $x<N_1$ then $y\leq x <N_1$ because $h\geq 1$.
Furthermore, 
\[g(y) h(x) = g(y)\max_{1\leq x',y'<N_1}\frac{g(x')}{g(y')}\geq g(y)\frac{g(x)}{g(y)}=g(x).\]
If $x\geq N_1$, let $i = i(x)\geq 1$, so that $x\geq N_i$ and $h(x) = x^{1/i}$.
Since $g$ is predictable, if $y\in[1,x/h(x))$ and $y+x\notin\delta_{\epsilon(x)}(g,x)$ then $g(y) \geq x^{-1/i}g(x)=g(x)/h(x)$.
\end{proof}

Let $r_\epsilon(x) := \max\{y\mid x+y'\in\delta_\epsilon(x), \text{ for all }|y'|\leq y\}.$
\begin{lemma}[Predictability Booster]\label{lem: predictability booster}
If $g$ is slow-dropping, slow-jumping, and predictable, then there is a sub-polynomial function $h$ such that for all $y\in[r_\epsilon(x)+1,x/h(x))$ we have $g(y)\geq g(x)/h(x)$.
\end{lemma}
\begin{proof}
According to Proposition~\ref{prop: predictable to sub-polynomial}, there exists a sub-polynomial function $h'$ such that for all $y\in[1,x/h'(x))$ if $x+y\notin\delta_{\epsilon/2}(x)$ then $g(y)\geq g(x)/h'(x)$.  
Without loss of generality $h'$ governs the slow-jumping of $g$ as well, hence $g(x)\leq (x/x')^2h'(x)g(x^\prime)$ for all $x'\leq x$.

We consider two cases: 
\begin{enumerate}
\item $|g(x)-g(x+r_{\epsilon}(x)+1)|>\epsilon g(x)$ and
\item $|g(x)-g(x-r_{\epsilon}(x)-1)|>\epsilon g(x)$.
\end{enumerate}

In the first case, by the definition of $r_{\epsilon}$ and predictability we have $g(r_{\epsilon}(x)+1)\geq g(x)/h'(x)$.
In the second case, let $x'=x-r_{\epsilon}-1$.  
If $g(x')\geq 2g(x)$ then $|g(x') - g(x)|> \epsilon g(x')$, as long as $\epsilon<1/2$.
Otherwise, $|g(x')-g(x)|\geq \epsilon g(x)\geq\frac{\epsilon}{2}g(x')$. 
Now, from predictability at $x'$ and the definition of $h'$, we find that 
\[g(r_{\epsilon}(x)+1)=g(x-x')
\geq \frac{g(x')}{h'(x')}
\geq \frac{g(x)}{4h'(x)^2},\]
where the last inequality follows because $g$ is slow-jumping and $x'\geq x/2$ (w.l.o.g., choose $h^\prime(x)>2$).

Now, since $g$ satisfies the slow dropping condition there exists a nondecreasing sub-polynomial function $h''$ such that $g(y')\geq g(r')/h''(y')$ for all $y'\in\N$ and $r'\leq y'$.
Thus, if $y\in [r_\epsilon(x)+1,x/h(x))$ 
\[g(y)\geq \frac{g(r_\epsilon(x)+1)}{h''(y)}\geq \frac{g(x)}{4h'(x)^2h''(x)},\]
so we take $h$ to be the product of $4(h')^2$ and $h''$.
\end{proof}

Recall the sub-polynomial function $H$ from the last section and let it also satisfy Lemma~\ref{lem: predictability booster} with $\epsilon$ replaced by $\epsilon/2$.
In particular $H$ now satisfies three conditions:
\begin{itemize}
\item $g(y)\geq g(x)/H(y)$, for all $x<y$, and
\item $g(y)\leq (y/x)^2 H(y)g(x)$, for all $x<y$. 
\item for all $y\in[r_{\epsilon/2}(x)+1,x/h(x))$ we have $g(y)\geq g(x)/H(x)$.
\end{itemize}

Let $y$ be the (additive) error in our estimate of frequency $x$ of a $(g,\lambda)$-heavy hitter. 
If $r_\epsilon(x)>|y|$ then we are fine.  
If not, consider the point $x+r_\epsilon(x)+1$ (or $x-r_\epsilon(x)-1$), which has $|g(x+r_\epsilon(x)+1)-g(x)|>\epsilon g(x)$ and therefore $g(r_\epsilon(x)+1)>g(x)/H(M)$.
Now, since $x$ is the frequency of a $(g,\lambda)$-heavy hitter it happens that $r_\epsilon(x)+1$, were it a frequency in the stream, would be $(g,\lambda/H(M))$-heavy and thus $\lambda/H(M)^2$-heavy for $F_2$.
Presuming that there are not too many items in the stream with frequency larger than $r_\epsilon(x)$, this implies that CountSketch produces an estimate for $x$ with error smaller than $r_\epsilon(x)$.
Now, since $g$ satisfies the slow dropping condition there cannot be too many frequencies larger than $r_\epsilon(x)$ in the stream because 
\[g(z) \geq \frac{g(r_\epsilon(x)+1)}{H(M)}\geq \frac{g(x)}{H(M)^2},\text{ for all }z>r_\epsilon,\]
which implies that there are at most $\frac{1}{\lambda}H(M)^2$ frequencies in this range.

\begin{algorithm}
\begin{algorithmic}[1]
\Procedure{1-Pass Heavy Hitters}{$g$, $\lambda$, $\epsilon$, $\delta$}
\State $\hat{S},\hat{V}\gets$CountSketch$(\frac{\lambda}{3H(M)}, \frac{\epsilon}{2H(M)},\delta/2)$
\State $\hat{F}_2\gets $AMS$(\epsilon,\delta/2)$
\State $S\gets\{i\in \hat{S}: |g(\hat{v}_i)-g(\hat{v}_i+y)|\leq \epsilon g(\hat{v}_i+y),$\\ $\qquad\quad\text{ for all }-\frac{\epsilon}{2H(M)}\sqrt{\hat{F}_2}\leq y \leq \frac{\epsilon}{2H(M)}\sqrt{\hat{F}_2}\}$
\State \Return $(j,\hat{v}_j)$, for all $j\in S$ 
\EndProcedure
\end{algorithmic}
\caption{A 1-pass $(g,\lambda,\epsilon,\delta)$-heavy hitters algorithm.}
\label{algo: 1-pass}
\end{algorithm}

\subsection{One Pass Lower Bounds}

The proof of the following theorem is broken up into Lemmas~\ref{lem:normslowdropping}, \ref{lem:normslowjumping} and \ref{lem:normpredictable}.

\begin{theorem}\label{thm:normintractable}
If $g\in\G$ is a 1-pass tractable $\bbS$-normal function, then $g$ is slow-jumping, slow-dropping, and predictable.
\end{theorem}
\begin{proof}
Immediate from Lemmas~\ref{lem:normslowdropping} , \ref{lem:normslowjumping}, and \ref{lem:normpredictable}.
\end{proof}

\begin{lemma}
\label{lem:normslowdropping}
If an $\bbS$-normal function $g\in\G$ is not slow-dropping, then $g$ is not $1$-pass tractable.
\end{lemma}
\begin{proof}
Suppose $g$ is not slow-dropping. Then there exist $0<\alpha\leq 1$ and integer sequences $y_1,y_2,\ldots\in\N$ and  $x_1,x_2,\ldots\in\N$, with $y_k$ increasing, such that $x_k<y_k$ and $g(x_k)\geq y_k^\alpha g(y_k)$, for all $k\geq 1$.
Furthermore, since $g$ is $\bbS$-normal we may choose the sequences so that there exists a sub-polynomial function $h$ satisfying
\begin{equation}\label{eq: difference for 1-pass slowdrop bound} |g(x_k+y_k)-g(x_k)|>  \frac{1}{h(y_k)}\min\{g(x_k),g(x_k+y_k)\},
\end{equation} for all $k$. 
We claim that one can take $|g(x_k+y_k)-g(x_k)|> g(x_k)/h(y_k)$.
If $g(x_k+y_k)<\frac{1}{2}g(x_k)$ then this is true because the constant function $2$ is sub-polynomial.  
On the other hand, if $g(x_k+y_k)\geq\frac{1}{2}g(x_k)$ then replacing $h$ by $2h$ in \eqref{eq: difference for 1-pass slowdrop bound} does the job.
Also, note that $1/h(x)$ is still a sub-polynomial function. 

Let $\mathcal{A}$ be a 1-pass $(g,\epsilon)$-SUM algorithm with $\epsilon=(3h(n))^{-1}$. 
We will show that $\calA$ uses $\Omega(n^\alpha)$ bits on a sequence of $g$-SUM problems with $n=y_1,y_2,\ldots$, hence $g$ is not tractable.
Let $n=y_k$ and $x=x_k$ and consider the following protocol for $\textsf{INDEX}(n^{\alpha})$ using $\calA$. 
Alice receives a set $A\subseteq[n^\alpha]$ and Bob receives an index $b\in[n^\alpha]$.
Alice and Bob jointly create a notional stream~$D$ and run $\calA$ on it.
Alice contributes $n$ copies of $i$ for each $i\in A$ to the stream and Bob contributes $x$ copies of his index~$b$. 
If $b\notin A$ there are $|A|$ items in the  stream with frequency $n$ and one with frequency $x$; whereas if $b\in A$ then $|A|-1$ frequencies are equal to $n$ and one is $n+x$.
Alice runs $\mathcal{A}$ on her portion of the stream and sends the memory to Bob along with the value $|A|$.
Bob completes the computation with his portion of the stream.
Let $\hat{G}$ be the value returned to Bob by $\calA$.
Bob decides that there is an intersection if $\frac{\hat{G}}{|A|g(n)+g(x)}\notin [1-\epsilon,1+\epsilon]$.

With probability at least $2/3$, $\hat{G}$ is a $(1\pm\epsilon)$-approximation to $g(V(D))$; suppose that this is so.
If $b\in A$ then $g(V(D))=(|A|-1)g(n) + g(x+n)$ and otherwise the result is $g(V(D))=|A|g(n) + g(x)$.
Without loss of generality $g(n)< \epsilon g(x)$, thus the difference between these values is 
\begin{align*}
|g(n)+g(x)-g(x+n)| &> |g(x)-g(x+n)|-\epsilon g(x)\\
&\geq 2\epsilon g(x)\\
&\geq \epsilon( g(x)+n^\alpha g(n))\\
&\geq \epsilon(g(x)+|A|g(n)).
\end{align*}
Thus Bob's decision is correct and $g$-SUM inherits the $\Omega(n^{\alpha})$ lower bound from $\textsf{INDEX}(n^\alpha)$.
\end{proof}

\begin{lemma}
\label{lem:normslowjumping}
If an $\bbS$-normal function $g\in\G$ is not slow-jumping, then $g$ is not $1$-pass tractable.
\end{lemma}
\begin{proof}
Since $g$ is not slow-jumping, there exist $\alpha >0$, a strictly increasing sequence $y_1,y_2,\ldots\in\N$, and a sequence $x_1,x_2,\ldots\in\N$ such that $x_k\leq y_k$ and $g(y_k) > s_k^{2+\alpha}x_k^\alpha g(x_k)$, where $s_k=\lfloor y_k/x_k\rfloor$.
Without loss of generality $y_k$ is strictly increasing.
According to Lemma~\ref{lem:normslowdropping}, we can assume that $g$ is slow-dropping, since otherwise it is not $1$-pass tractable. 
Hence, there exists $N\in\N$ such that for all $x\ge N$ and $r\leq x$ we have  $\frac{1}{2}x^{\alpha}g(x)\ge g(r)$. 
If $x_k\leq N$ for all but finitely many $k$, we may assume this holds for all $k$.  Otherwise, by taking a subsequence, we can assume $x_k>N$ holds for all $k$.

Let $r_k = y_k-s_kx_k$ be the sequence remainders.
Before proceeding with the reduction we will establish the bound $2g(r_k)\leq g(y_k)$.
If $x_k\leq N$, for all $k$, then $s_k$ is unbounded while $x_kg(x_k)$ is bounded away from zero.
This means that $g(y_k)$ is unbounded while $g(r_k)$ is bounded and we can assume $y_1$ is large enough that $g(y_k)\geq 2g(r_k)$ holds for all $k$.
Otherwise $x_k>N$, for every $k$, hence $2g(r_k)\leq x_k^{\alpha}g(x_k)\leq g(y_k)$, for each $k$, from the definition of $N$.

Let $\calA$ be a $(g,\epsilon)$-SUM algorithm with accuracy $\epsilon=1/12$.
Consider the $\textsf{DISJ+IND}(n,t)$ problem, where $t=s_k$ and $n= s_k^{2+\alpha}x_k^{\alpha}$. 
Denote $x:=x_k$, $y:=y_k$, $s:=s_k$ and $r:=r_k$. The $t+1$ players receive sets $A_1,A_2,\ldots,A_t\subseteq[n]$ and an index $b\in[n]$.
As before, the players jointly run $\calA$ on a notional stream and share $|A_i|$. 
Each player~$i$ places $x$ copies of each $j\in A_i$ into the stream except for the final player who places $r$ copies of his index $b$. 
Let $n^\prime=\sum_i |A_i|$ be the total size of the $t$ players' sets. 
On an intersecting instance the result of $g$-SUM is $a_1:= (n^\prime-t)g(x)+g(y)$, and on a disjoint instance the result is $a_2:= n^\prime g(x) + g(r)$.

The difference is 
\begin{align}
a_1-a_2 &\geq g(y)-g(r) - tg(x)>\frac{1}{2}g(y)-tg(x) \nonumber\\
&\geq\frac{1}{6}(g(y)+(2s^{2+\alpha}x^{\alpha})g(x)) -sg(x) \nonumber\\
&=\frac{1}{6}(g(y)+(s^{2+\alpha}x^{\alpha} + s^{2+\alpha}x^{\alpha} - 6s)g(x))\nonumber\\
&\geq  2\epsilon(g(y) + s^{2+\alpha}x^{\alpha}g(x))\nonumber\\
&= 2\epsilon(n'g(x)+g(y)),
\end{align}
where the last inequality holds for all sufficiently large $n$. 
Thus, the index player can correctly distinguish which case has occurred, and the algorithm requires $\Omega(n/t^2)=\Omega(y^\alpha)$ bits.
\end{proof}

\begin{lemma}
\label{lem:normpredictable}
If an $\bbS$-normal function $g\in\G$ is not predictable, then $g$ is not $1$-pass tractable.
\end{lemma}
\begin{proof}
Since $g$ is not predictable, there exists a sub-polynomial function $\epsilon$ and constant $\gamma>0$ such that some infinite sequence $x_k,y_k$ satisfies the following:
\begin{itemize}
\item $x_k\to\infty$ and $y_k\in[1,x_k^{1-\gamma})$,
\item $x_k+y_k\notin\delta_{\epsilon(x_k)}(g,x_k)$, and
\item $x_k^{\gamma}g(y_k)<g(x_k)$.
\end{itemize}
The proof is by a reduction from $\textsf{INDEX}(n)$ with $n=\epsilon(x_k)x_k^\gamma/4$.
Suppose $\calA$ is an algorithm for $(g,\epsilon/4)$-SUM.
Alice receives a set $A\subseteq[n]$ and Bob receives an index $b\in[n]$.
Alice adds $y_k$ copies of $i$ to the stream, for each $i\in A$, runs $\calA$ on her portion of the stream, and sends the contents of the memory to Bob.
Bob adds $x_k$ copies of $b$ to the stream and completes the computation.

The stream has $|A|$ frequencies equal to $y_k$ and one equal to $x_k$, if there is no intersection, or $|A|-1$ equal to $y_k$ and one equal to $x_k+y_k$, if there is an intersection.
Recall that $x_k+y_k\notin\delta_\epsilon(g,x_k)$, hence $|g(x_k)-g(x_k+y_k)|>\epsilon g(x_k)$, and by construction \[|A|g(y_k)\leq \frac{\epsilon(x_k)}{4}x_k^\gamma g(y_k)\leq\frac{\epsilon(x_k)}{4}g(x_k).\]
Therefore, Bob can correctly distinguish whether $b\in A$ when $\calA$ yields a $(1\pm\epsilon/4)$-approximation.
Thus $(g,\epsilon/4)$-SUM requires $\Omega(n)$ bits.
\end{proof}

\subsection{Two Pass Lower Bounds}

\begin{theorem}\label{thm:norm2passbound}
If a $\bbP$-normal function $g\in \G$ is tractable then $g$ is slow-dropping and slow-jumping.
\end{theorem}

\begin{proof}
Immediate from Lemmas \ref{lem:pnormalslowdropping} and \ref{lem:pnormalslowdjumping}.
\end{proof}

\begin{lemma}\label{lem:pnormalslowdropping}
If a $\bbP$-normal function $g\in\G$ is not slow-dropping, then $g$ is not $O(1)$-pass tractable.
\end{lemma}
\begin{proof}
Suppose $\mathcal{A}$ is a $p=O(1)$ pass algorithm that solves $(g,\epsilon)$-SUM and $g$ does not satisfy slow-dropping condition. 
Then there exists $\alpha > 0$ such that for any $N>0$, there exists $y>N$ and $x<y$ satisfying $g(y) < g(x) / y^\alpha$. 
Separately, since $g$ is $\bbP$-normal, there exists $\beta>0$ such that for any $N>0$ if $g$ has an $\alpha$-period $y>N$ then there is an $\alpha$-period $y>N$ and $x<y$ such that $g(y) \le g(x) / y^\alpha$ and $|g(x+y)-g(x)|> y^\beta \min(g(x), g(x+y)) $. 

Let $\gamma = \min(\alpha, \beta)$. 
Consider the following protocol for DISJ$(n,2)$, where $n=y^{\gamma/2}$.

First, consider $g(x+y) \le g(x)$. Since  $|g(x+y) - g(x)| > y^\gamma g(x+y)$ and $y^\gamma g(x+y) \ge g(x+y)$, then $g(x)\ge y^\gamma g(x+y)$. 
The players jointly create a stream where Player 1 inserts $x$ copies of each element of her set $S_1$ into the stream, and Player 2 inserts $y$ copies of every element $a$ not in her set $S_2$. 
First Player 1 creates her portion of the stream, runs the first pass of $\calA$ on it, and sends the memory to Player 2.  
She completes the first pass with her portion of the stream, and returns the memory to Player 1.
The players continue in this way for a total of $p$ passes over the stream.

Let $V$ denote the frequency vector of the resulting stream.
If there is an intersection, let $S_1\cap S_2 = \{a\}$.
Then $S_1 \cap \bar{S}_2 = S_1\backslash \{a\}$, and $g(V)$ is $r_1 = (|S_1|-1)g(x+y) + (n-|S_2|-|S_1|)g(y) +g(x) + g(y)$. 
If there is no intersection, the value of $g(V)$ is $r_2 = (|S_1|-1)g(x+y) + (n-|S_2|-|S_1|)g(y) + g(x+y)$. 
Notice that
\[\frac{|r_2-r_1|}{r_1} \ge \frac{|g(x+y)-g(x)|}{g(x)} \geq \frac{1}{2},\] 
for sufficiently large $n$. 
For any  $\epsilon < 1/2$, $\mathcal{A}$  is able to distinguish the 2 cases, which gives a lower bound on the memory bits used by $\mathcal{A}$ is $\Omega(n/p)$.

The case $g(x+y) > g(x)$ is the same except Player 2 inserts $y$ copies of each element that \emph{is} in her set $S_2$. 
Thus, $\mathcal{A}$ uses $\Omega(n/p)$ bits of memory, hence $g$ is not $O(1)$-pass tractable.
\end{proof}

\begin{lemma}
\label{lem:pnormalslowdjumping}
If a $\bbP$-normal function $g\in\G$ is not slow-jumping, then $g$ is not $O(1)$-pass tractable.
\end{lemma}
\begin{proof}
Suppose $\mathcal{A}$ is a $p=O(1)$ pass algorithm for $(g,\epsilon)$-SUM.
If $g$ is not slow-jumping, then there exists $\alpha >0$ such that for any $N >0$, there exists $x<y\in \mathbb{N}$ and $y\ge N$ but $g(y) > \lfloor y/x\rfloor^{2+\alpha}x^\alpha g(x)$.
Notice that for $y>x$ we have $\lfloor y/x\rfloor \geq y/2x$, so by adjusting $\alpha$ we can assume $g(y)>(y/x)^2y^\alpha g(x)$.

We can further assume $g$ is slow dropping, since otherwise $g$ is not $O(1)$-pass tractable by Lemma~\ref{lem:pnormalslowdropping}. 
Thus, there exists a nondecreasing sub-polynomial function $h(x)>1$ such that $g(x)\leq h(y) g(y)$.

Consider an instance $A_1,A_2,\ldots,A_t\subseteq[n]$ of $\textsf{DISJ}(n,t)$, where $t=\lceil y/x\rceil$ and $n= (\frac{y}{x})^2\frac{y^\alpha}{2h(y)}$. 
Each of the first $t-1$ players, inserts $x$ copies of her elements into the stream and the $t$th player inserts $y-(t-1)x<x$ copies of her elements into the stream.
The players pass the memory and repeat to run the $p$ passes of $\calA$ as before.
Notice that $g(y-(t-1)x)\leq h(x)g(x)$ by the slow dropping condition.

Let $n^\prime=\sum_{i}|A_i|$, and let $V$ be the frequency vector of the stream created. 
If there is no intersection, the value of $g(V)$ is
\begin{align*}
(n'-|A_t|)g(x) &+ |A_t|g(y-(t-1)x)\leq n'h(x)g(x)\\
&\leq nh(x)g(x) \leq \frac{1}{2}g(y).
\end{align*}
On the other hand, if there is an intersection then $y$ is the frequency of some item in the stream, and therefore $g(V)\geq g(y)$.
Thus $\calA$ distinguishes between the cases.
The algorithm uses $\Omega(n/pt^2)=\Omega(y^{\alpha/2}/p)$ bits of memory, which proves that $g$ is not $p$-pass tractable.
\end{proof}

\subsection{Zero-One Law Proofs}\label{sec:zoproofs}
\begin{customthm}{\ref{thm:norm1passalg}}[1-pass Zero-One Law]
A function $g\in\G$ is 1-pass tractable and $\bbS$-normal if and only if it is slow-jumping, slow-dropping, and predictable.
\end{customthm}
\begin{proof}
The lower bound is proved as Theorem~\ref{thm:norm2passbound}, hence the algorithm remains.
With the Recursive Sketch of Theorem~\ref{thm: recursive sketches}, it is enough to show that Algorithm~\ref{algo: 1-pass} is a $(g,\lambda,\epsilon,\delta)$-heavy hitters algorithm for $\lambda = \epsilon^2/\log^3{n}$ and $\delta=1/\log{n}$.

By Lemma~\ref{lem: identifying heavy hitters}, the Count Sketch used by Algorithm~\ref{algo: 1-pass} returns a list of pairs $(i_j,\hat{v}_{i_j})$ containing all of the $\lambda/3H(M)$-heavy elements for $F_2$.
By definition, any item $i_j$ that survives the pruning stage has $|g(\hat{v}_{i_j})-g(v_j)|\leq\epsilon g(v_j)$.
Hence, it only remains to show that every $(g,\lambda)$-heavy hitter survives the pruning stage.

Suppose $i'$ is the index of the $(g,\lambda)$-heavy hitter that minimizes $r_{\epsilon/2}(v_{i'})$.
Now suppose that we insert a new item~$i''$ in the stream with frequency $v_{i''}=r_{\epsilon/2}(v_{i'})+1$.
Then
\[\frac{g(v_{i'})}{H(M)}\leq g(v_{i''})\leq H(M)g(v_{i'}),\]
where the first inequality follows from Lemma~\ref{lem: predictability booster} and the second because $g$ is slow-dropping.
Thus, this item is a $(g,\frac{\lambda}{3H(M)})$-heavy hitter and the constant term on the parameter $b$ for the Count Sketch can be chosen to guarantee additive error no more than $v_{i''}/3\leq r_{\epsilon/2}(v_{i'})/3$.
The same guarantee holds for the stream without $i''$, thus for every $(g,\lambda)$-heavy hitter~$i_j$ we have 
\[|\hat{v}_{i_j}-v_{i_j}|\leq \frac{1}{3}r_{\epsilon/2}(v_{i'})\leq \frac{1}{3} r_{\epsilon/2}(v_{i_j}),\]
where the last inequality holds by the choice of $i'$.
Furthermore, for every $-r_{\epsilon/2}(v_{i'})/3\leq y\leq r_{\epsilon/2}(v_{i'})/3$ we have $|(\hat{v}_{i_j}+y) - v_{i_j}|\leq r_{\epsilon/2}(v_{i_j})$ and thus
\begin{align*}
|g(\hat{v}_{i_j}) - g(\hat{v}_{i_j}+y)| \leq& |g(\hat{v}_{i_j}) - g(v_{i_j})|\\ 
&+ |g(v_{i_j}) - g(\hat{v}_{i_j}+y)|\\
\leq& \epsilon g(\hat{v}_{i_j}+y).
\end{align*}

With probability at least $1-2\delta/2=1-\delta$ both the Count Sketch and the AMS approximation of $F_2$ meet their obligations, in this case the output is correct. By the guarantee of CountSketch, we can safely assume that $r_{\epsilon/2}(v_i^\prime)\ge\epsilon/2H(M)\sqrt{\hat{F}_2}$, heavy hitters will survive the pruning stage.
\end{proof}

\begin{customthm}{\ref{thm:norm2passalg}}[2-pass Zero-One Law]
A function $g\in\G$ is 2-pass tractable and $\bbP$-normal if and only if it is slow-dropping and slow-jumping.
Furthermore, every slow-dropping and slow-jumping $\bbS$-normal function is also $2$-pass tractable.
\end{customthm}
\begin{proof}
The lower bound is proved as Theorem~\ref{thm:norm2passbound}.
The upper bound is governed by Proposition~\ref{prop:pnorm in snorm} and Theorem~\ref{thm:snorm2passalg}.
\end{proof}

Here are a few examples. 
The functions $x^2\lg(1+x)$, $\left(2+\sin\log(1+x)\right)x^2$, and $e^{\log^{1/2} (1+x)}$ are all 1-pass tractable because they are all slow-dropping, slow-jumping, and predictable.
On the other hand, $1/x$ is not slow-dropping, $x^3$ is not slow-jumping, and $\left(2+\sin\sqrt{x}\right)x^2$ is not predictable, so none of these functions is 1-pass tractable.
The last of the three is, however, slow-jumping and slow-dropping, and hence it is 2-pass tractable. 

\section{The Exotic Nature of Nearly Periodic Functions}\label{sec: exotic}

Nearly periodic functions are highly constrained to \emph{almost} repeat themselves like a periodic function.
These functions admit large changes in value that would imply a large lower bound on their space complexities, if they did not satisfy the many constraints. 

Appendix~\ref{app:nearlyperiodic} describes some of properties of $\bbS$-nearly periodic functions and gives an example of a tractable nearly periodic function. 
To illustrate the exotic nature of these functions we provide two properties of these functions. These properties suggest that the nearly periodic functions might not have many applications. 
The details of the following results and additional properties of nearly periodic functions are provided in the Appendix.

\begin{restatable}{proposition}{strangeLimits}
\label{prop: strange limits}
Let $g>0$ be an $\bbS$-nearly periodic function that is bounded above by a sub-polynomial function.
There exists $\alpha>0$ and an increasing sequence $y_k\in \N$ such that $g(y_k)\leq y_k^{-\alpha}$ and, for every $x\in\N$,  $\lim_{k\to\infty}g(x+y_k) = g(x).$
\end{restatable}

If an $\bbS$-nearly periodic function $g$ grows polynomially, then it can be made normal and intractable by introducing a small change. 
Let $L_\eta(g)(x) := g(x)\log^\eta(1+x)$.
The transformation $L_\eta$, for $\eta>0$, separates $\bbS$-normal 1-pass tractable functions from $\bbS$-nearly periodic functions in the following sense: For any 1-pass tractable $\bbS$-normal function $g$ and $\eta>0$ the function $L_\eta(g)$ is also 1-pass tractable, however for any $\bbS$-nearly periodic 1-pass tractable function and $\eta>0$ the function $L_\eta(g)$ is 1-pass intractable.

\begin{restatable}{theorem}{smallChangePeriodic}
\label{thm:small_change_periodic}
An $\bbS$-nearly periodic function $g$ is either $1$-pass intractable or $L_\eta(g)$ is $1$-pass intractable for any $\eta > 0$.
\end{restatable}

\begin{restatable}{theorem}{smallChangeOPass}
\label{thm:small_change_1pass}
If an $\bbS$-normal function $g$ is $1$-pass tractable, then $L_\eta(g)$ is $\bbS$-normal and $1$-pass tractable for any $\eta\ge 0$.
\end{restatable}

{
The set of tractable nearly periodic function can be further restricted using the {\sf ShortLinearCombination}
problem discussed in Section \ref{sec:intro}. In particular we show the following,
\begin{theorem}[Informal]
If the periods of a nearly periodic function $g$ form what we call an \emph{indistinguishable frequency set}, then any algorithm for the $g$-sum problem uses polynomial space.
\end{theorem}
The formal definition and proofs are presented in Theorem \ref{thm:communication nearly periodic}. We also prove our lower bounds for {\sf ShortLinearCombination} in Appendix \ref{app:communication} based on the information complexity framework. 
} We note that it is non-trivial to construct an example of a tractable nearly periodic function. 
We provide such an example in Appendix~\ref{app:nearlyperiodic}.

\bibliographystyle{plain}

\bibliography{ref}

\clearpage

\begin{appendix}

\section{The Case of $g(0)\neq 0$}\label{app:caseg00}
When $g(0)=0$, we have to re-address the lower bound of using INDEX since the elements not contributing to the frequency vector still contribute to the $g$-SUM. Therefore, we re-define the nearly periodic functions in the next section. We further show that functions crossing the axis are not $1$-pass tractable and functions with zero points are not tractable unless they are periodic in Section \ref{sec:crossing the axis}. We prove the same $1$-pass zero one law in the remaining 2 sections. The proof turns out to be very similar to the $g(0)=0$ cases with small number of additional tweaks. Note that we only provide the turnstile model lower bounds.

\subsection{Redefinition of Nearly Periodic}

For the $g(0)\neq 0$ case, we have the following definition of nearly-periodic. 
\begin{definition}
Given a set of functions $\calS$, call
$g(x)$ {\em $\calS$-nearly periodic}, if the following two conditions are satisfied.
\begin{enumerate}
\item There exists $\alpha >0$ such that for any constant $N>0$ there exists $x, y\in \mathbb{N}$, $x<y$ and $y\ge N$ such that  $g(y) \le g(x) / y^\alpha$. Call such a $y$ {\em $\alpha$-period of $g$};
\item For any $\alpha > 0$ and any error function $h\in \calS$ there exists $N_1>0$ such that for all $\alpha$-periods $y\ge N_1$ all $x<y$ such that $g(y) y^\alpha \le g(x)$ we have $|g(x)-g(x-2y)|\le \min\{g(x), g(x-2y)\}h(y).$
\end{enumerate}
A function $g$ is {\em $\calS$-normal} if it is not $\calS$-nearly periodic.  
\end{definition}

Let all functions be $\G_0^*=\{g:\mathbb{N}\rightarrow \mathbb{R}^+, g(x)=g(-x), \text{and}~g(0)=1\}$.

\subsection{Crossing the Axis
\label{sec:crossing the axis}}
If $g\in\G_0^*$ and there exist $x,y\in\mathbb{N}$ such that $g(x)>0>g(y)$, we have the following lemma.

\begin{lemma}
\label{lem:g1 less 0}
Let $g\in\G_0^*$, if $g(1)<0$, then any algorithm solves $g$-SUM requires $\Omega(n)$ space.
\end{lemma}
\begin{proof}
Consider the following INDEX$(n)$ reduction. Alice and Bob jointly create a stream as following. Let $A$ be Alice's set. $n_1 = |A|$. Let $n_2=\lfloor\frac{n-n_1}{-g(1)}\rfloor$. Let $C=(n_1+n_2)g(1) + (n-n_1)$. We have $0\le C <g(1)$. Alice choose a domain $[n+n_2]$ for the algorithm and output $1$ copy of each of her element. Bob receives the memory content of the algorithm from Alice as well as $n, n_1, n_2$. Bob outputs $-1$ copy of his index and also adds to the stream one copy of each of the following elements: $n+1, n+2, \ldots, n+n_2$.

If there is an intersection, the result of $g$-SUM is 
\[
(n-n_1) + (n_1+n_2 - 1) g(1) + g(0) = C + g(0) - g(1);
\]
If there is no intersection, the result is 
\[
(n-n_1 - 1) + (n_1 + n_2) g(1) + g(1) = C + g(1) - g(0).
\]
Then a constant approximation algorithm can distinguish the two cases. The algorithm inherits an $\Omega(n)$ bound from INDEX.
\end{proof}

\begin{lemma}
\label{lem:gxgy}
Let $g\in \G^*_0$ and $g(1) > 0$, if there exists $y\in\mathbb{N}$ such that $g(y)<0<g(1)$, then there exits $z\in\mathbb{N}$ such that $g(1+z)\neq g(1-z)$.
\end{lemma}
\begin{proof}
If for every $z$, $g(1+z)=g(1-z)$, we have $g(2)=g(0)=1$ and for any $w$, $g(w+2)=g(w)$. Therefore $2$ is a period of the function. Now, consider the following cases, 1) if $y$ even, then $g(y)=g(2 \lfloor y/2\rfloor))=g(0)$ ($2$ is a period), which is a contradiction that $g(y)<0$; 2) if $y$ odd, then $g(y) = g(2 \lfloor y/2\rfloor + 1) = g(1)>0$, contradicting $g(y) <0$.
\end{proof}

\begin{proposition}
Let $g\in\G_0^*$ and $g(1) >0$, if there exists $y\in\mathbb{N}$ such that $g(y)<0$, then any algorithm solves $g$-SUM requires $\Omega(n)$ space.
\end{proposition}
\begin{proof}
Consider the following reduction from INDEX$(n)$. Let $n^\prime=\lfloor -(n-1)g(1)/g(y)\rfloor$ and $z\in N$ be the integer given by Lemma~\ref{lem:gxgy}. Let $C=(n-1)g(1) + n^\prime g(y)$. We have $0<C<-g(y)$. Alice and Bob jointly create a data stream in domain $n+n^\prime$: Alice outputs $1$ copy of every element in her set, and $-1$ copy of every element not in her set. Bob outputs $y$ copy of each of the elements $i+1, i+2, \ldots, i+n^\prime$ and $z$ copy of his index. If there is no intersection, the result of $g$-SUM is 
\[
(n-1)g(x) + n^\prime g(y)+g(z-1) = C+g(z-1).
\]
If there is an intersection, the result of $g$-SUM is 
\[
(n-1)g(x) + n^\prime g(y) + g(z+1) = C+g(z+1)
\]
Therefore any constant approximation algorithm can distinguish the two cases, implying an $\Omega(n)$ bound for the memory of the algorithm.
\end{proof}

\begin{proposition}
Let $g\in\G_0^*$ if $g(x)=0$ and $g(2x)\neq g(0)$  then $g$ is not $1$-pass tractable.
\end{proposition}
\begin{proof}
Consider the INDEX$(n)$ reduction. Alice and Bob jointly create a stream as following, Alice output $x$ for every element in her set, $-x$ for every element not in her set. Bob output $x$ for his index. If there is an intersection, the $g$-SUM result is $g(2x)$. If there is no intersection, the result is $g(0)$. 
\end{proof}

\begin{proposition}
If function $g\in \G_0^*$ with $g(x) =0$ for some $x\in \mathbb{N}$ is tractable, then $g$ is periodic.
\end{proposition}
\begin{proof}
By reduction from INDEX$(n)$, Alice output $x$ for each of her elements and $-x$ for each of her non-elements. Bob output $k+x$ for his index. If there is an intersection, the $g$-SUM result is $g(k+2x)$. If there is no intersection, $g$-SUM is $g(k)$. If $g(k+2x)\neq g(k)$ for any $k$, a constant approximation algorithm can distinguish the 2 cases. Therefore, $2x$ is a period of the function.
\end{proof}

Therefore, we only consider the following set of functions,
\[\G_0=\{g:\mathbb{N}\rightarrow \mathbb{R}^+, g(x)=g(-x)>0, \text{and}~g(0)=1\}\]

\subsection{Lower Bounds}

\begin{theorem}
If an $\bbS$-normal function $g\in\G_0$ does not satisfy slow dropping, then $g$ is not $1$-pass tractable.
\end{theorem}
\begin{proof}
The reduction to INDEX$(n^\alpha)$ is similar to the case of $g(0)=0$. Now Alice outputs $n$ copies of the elements in $A$ and also outputs $-n$ copies of $i$ if $i\not\in A$. Bob output $x-n$ copies of his index. If there is an intersection the $g$-SUM result is $g(V(D)) = (n^\alpha-1)g(n)+g(x)$ and otherwise is $g(V(D)) = (n^\alpha-1)g(n)+g(x-2n)$. Since $g$ is not nearly periodic, the algorithm is able to distinguish the two cases.
\end{proof}

\begin{theorem}
If an $\bbS$-normal function $g\in\G_0$ is not slow-jumping, then $g$ is not $1$-pass tractable. 
\end{theorem}
\begin{proof}
The reduction to DISJ+IND is the same to the case of $g(0)=0$. 

Now by slow dropping, there exists a sub-polynomial function $h(x)$ such that $g(0)\le g(x)h(x)$. The $g$-SUM result differs from $g(0)=0$ case by $(n-n^\prime)g(0)$ (or $(n-n^\prime + t)g(0)$), which is comparable to $(n-n^\prime)g(x)$. Therefore, we can apply the same analysis.
\end{proof}

\begin{theorem}
If an $\bbS$-normal function $g\in\G_0$ is not predictable, then $g$ is not $1$-pass tractable.
\end{theorem}
\begin{proof}
Still use the same reduction to INDEX with the $g(0)=0$ case. Now the $g$-SUM results differs by $(n-|A|)g(0)$. By slow dropping, there exists a sub-polynomial function $h(x)$ such that $g(0)\le h(y)g(y)$. Therefore, all the analyses remain the same.
\end{proof}

\subsection{Upper Bound}
The algorithm should still work by the following 2 lemmas.
\begin{lemma}\label{lem: g heavy implies F2 heavy g00}
Let $g\in\G_0$ be an $\bbS$-normal function that is slow-jumping and slow-dropping.
There exists a sub-polynomial function $h$ such that for any $D\in\calD(n,m)$ with frequencies $v_1,v_2,\ldots,v_n$, if $g(v_i)\geq \lambda \sum_j g(v_j)$ then $v_i^2\geq \frac{\lambda}{h(|v_i|)}\sum_{|v_j|<|v_i|} v_j^2$.
\end{lemma}
\begin{proof}
The proof is the same as in the $g(0)=0$ case with additional care of $g(0)=1$.
\end{proof}

\begin{lemma}\label{lem: predictability booster of g00}
If $g\in\G_0$ is slow-dropping, slow-jumping, and predictable, then there is a sub-polynomial function $h$ such that for all $y\in[r_\epsilon(x)+1,x/h(x))$ we have $g(y)\geq g(x)/h(x)$.
\end{lemma}
\begin{proof}
The proof is the same as in the $g(0)=0$ case with additional care of $g(0)=1$.
\end{proof}

\section{Lower bound for $\disjind$ }\label{app:complexity of disjind}
\begin{theorem}\label{thm: communication complexity of disjind}
The randomized one-way communication complexity of $\disjind(n,t)$ is $\Omega(n/t\log{n})$.
\end{theorem}

\begin{proof}
We give a reduction from $\disj(n,t+1)$.  
Let $\calP$ be any randomized protocol for $\disjind(n,t)$.
Run in parallel $\ell=\lceil 96\log n\rceil$ independent copies of $\calP$ through the first $t$ players. 
This produces $\ell$ transcripts.
Player $t+1$ now takes the each of the transcripts and computes the final value of each once for every element he holds, as if it was the only element he held.
No communication is need for this part because $\calP$ is a one-way protocol and $t+1$ is the final player.

Player $t+1$ then takes the majority vote among the independent copies of $\calP$ for each element.
If any vote signals an intersection then he reports intersection; otherwise he reports disjoint.

If every one of the $|A_{t+1}|\leq n$ majorities is correct then the final player's report is correct.
Let $X_i$, for $i\in A_{t+1}$, be the number among the $\ell$ copies of $\calP$ with the correct outcome when the final player completes protocol using $i\in A_{t+1}$. 
Then $X_i$ is Binomially distributed from $\ell$ trials with success probability at least 2/3.
Using a Chernoff Bound we find
\begin{align*}
P(X_i\leq \ell/2) 
&= P\left(X_i \leq (1-\frac{1}{4})\frac{2}{3}\ell\right) 
\leq \exp\left\{\frac{-1}{32}\mu\right\}\\
&\leq \exp\left\{\frac{-1}{32}\cdot\frac{2}{3}\ell\right\}
\leq \frac{1}{n^2}.
\end{align*}
Thus, with probability at least $1-\frac{1}{n}$ every majority vote is correct, hence our $\disj(n,t)$ protocol is correct for $n\geq 3$.

Let $T_1,T_2,\ldots,T_\ell$ be the transcripts.
The total cost of this $\disj$ protocol is $\sum_{i=1}^\ell|T_i|$.
Since $\disj(n,t+1)$ requires $\Omega(n/t)$ bits of communication, at least one of the protocols has length $\Omega(n/t\ell)=\Omega(n/t\log{n})$, hence $|\calP|= \Omega(n/t\log{n})$ bits of communication.
\end{proof}

\section{Lower bound for {\sf ShortLinearCombination}}\label{app:communication}
\subsection{The 3-frequency Distinguishing Problem}
We first define the {\sf ShortLinearCombination} problem for three frequencies, which we call $(a,b,c)$-DIST for short.
\begin{definition}
A stream $S$ with frequency vector $v$ is given to a one-pass streaming algorithm. $v$ is promised to be from $V_0=\{-a, a, -b, b, 0\}^n$ or $V_1=\{$EMB$(v, i, e)\mid v\in V_0, i\in [n], e\in\{-c,c\}\}$, where $a,b,c>0$ and \[
\text{EMB}(v, i, e)_j = \begin{cases}
v_i & i\neq j\\
e & i=j
\end{cases}
\]
(in other words, $V_1$ is given by replacing a coordinate of a vector $v$ from $V_0$ with $-c$ or $c$). The $(a,b,c)$-DIST problem is for the algorithm to distinguish whether $v\in V_0$ or $v\in V_1$.
\end{definition}

To study this problem, we first consider the special case where $c=\gcd(a,b)=1$. 
\begin{proposition}
If $\gcd(a,b) = 1$, then any randomized algorithm solving $(a,b,1)$-DIST with probability at least $2/3$ requires $\Omega(n/\max(a,b)^2)$ bits of memory.
\end{proposition}
\begin{proof}
We use the same notation as in \cite{bar2002information}, and in particular
refer to that work for background on information complexity. 
Without loss of generality, we assume $a>b$.
Consider a communication problem in which Alice is given a vector $v_1$ and Bob is given a vector $v_2$, for which $v = v_1 - v_2 \in V_0\cup V_1$. The players are asked to solve the $(a,b,1)$-DIST on $v$. This problem is OR-decomposable with primitive DIST$(x, y)$, where DIST$(x,y)=1$ if and only if $|x-y| = 1$.  Consider a randomized protocol $\Pi$ which succeeds with probability at least $2/3$ in solving this problem, and suppose $\Pi$ has the minimum communication cost of all such protocols. We now call $(a,b,1)$-DIST $f$ for short. The following is well known (see, e.g., \cite{bar2002information}), 
\[
R_{2/3}(f) \ge IC_{\mu, 2/3}(f),
\]
where $R_{2/3}(f)$ is the communication complexity of the best randomized protocol (with error probability at most $2/3$) on the worst case input, $\mu$ is any distribution over the input space, and $IC_\mu$ is the information complexity of the best protocol over the input distribution $\mu$. We now construct an input distribution $\mu$ as follows. Let $D$ be chosen uniformly at random over $\{$Alice, Bob$\}$, $E$ is chosen uniformly from $\{a, a+1, a+2, \ldots, m-a\}$, where $m\gg a$ is a sufficiently large integer for which the coordinates of the vector jointly created by Alice and Bob will never exceed $m$. Denote by $\xi$ the joint distribution of $(D, E)$. Based on the value of $D$ and $E$, the input distribution $\gamma$ of Alice and Bob is decided as follows: if $D$ chooses Alice, then the input $X$ of Alice is chosen uniformly at random from the multiset $\{E-a, E-b, E, E, E, E, E+a, E+b\}$, and Bob is given $E$; if $D$ chooses Bob, $\gamma$ just swaps the role of Alice and Bob.
Clearly, $\gamma$ is a product distribution conditioned on $(D, E)$. We have that $\zeta = (\gamma, \xi)$ is a mixture of product distributions. Let $\eta = \zeta^n$, and let $\mu$ be the distribution obtained by marginalizing $(\bf{D}, \bf{E})$ from $\eta$. Since every joint vector created above is from $V_0$, $\mu$ is a {\it collapsing distribution} for $f$ (see \cite{bar2002information}). Then by Lemma 5.1 and Lemma 5.5 of \cite{bar2002information}, we have the following direct sum result:
\[
I(\mathbf{X}, \mathbf{Y}; \Pi(\mathbf{X}, \mathbf{Y})\mid \mathbf{D}, \mathbf{E})\ge n \cdot \text{CIC}_{\zeta, 2/3}(\text{DIST}).
\]
We now expand the terms of CIC$_{\zeta, 2/3}($DIST$)$. Let $[a, m-a] = \{a, a+1, a+2, \ldots, m-a\}$, $U_e$ be the random number from the uniform distribution on multiset $K\equiv\{e-a, e+a, e, e, e, e, e-b, e+b\}$, and $\Psi$ be the transcript of a communication-optimal randomized protocol for the primitive function DIST. Let $\beta_{e,y}$ be the distribution of $\Psi(e,y)$ and $\beta_e = \frac{1}{8}\sum_{y\in K}\beta_{e,y}$ be the average distribution. The following inequalities follow from Lemma 2.45 and Lemma 2.52 in Bar-Yossef's PhD thesis \cite{bar2002complexity}, the Cauchy-Schwarz inequality, and the triangle inequality.

\begin{align}
\text{CIC}&_{\zeta, 2/3}(\text{DIST})  \nonumber \\
=&\frac{1}{2|m-2a|}\sum_{e\in[a, m-a]} [I(U_e; \Psi(U_e, e))\nonumber \\
 &\qquad\qquad + I(U_e; \Psi(e, U_e))]\nonumber \\
=&\sum_{\substack{e\in[a, m-a]\\ e_j\in K}} \frac{Pr[U_e = e_j]}{2(m-2a)} [ KL(\beta_{e, e_j} || \beta_e) \nonumber \\
&\qquad\qquad+ KL(\beta_{e_j, e} || \beta_e)]\nonumber \\
&\text{(Bar-Yossef's Thesis Proposition 2.45)} \nonumber
\\
\ge& \frac{1}{8\ln 2(m-2a)}\sum_{\substack{e\in[a, m-a]\\j\in \{-a, -b, 0,..0, a, b\}}} [h^2(\beta_{e+j, e}, \beta_e)\nonumber \\
&\qquad\qquad + h^2(\beta_{e, e+j}, \beta_e)] 
\nonumber \\
&\text{(Cauchy-Schwarz)} \nonumber
\\
\ge &\frac{1}{16\ln 2(m-2a)}\sum_{\substack{e\in[a, m-a]}} [\left(h(\beta_{e, e}, \beta_e) + h(\beta_{e, e+a}, \beta_e)\right)^2\nonumber\\
& + (h(\beta_{e-a, e}, \beta_e) + h(\beta_{e, e}, \beta_e))^2\nonumber \\
&+\left(h(\beta_{e, e}, \beta_e) + h(\beta_{e, e+b}, \beta_e)\right)^2\nonumber\\ 
&+ (h(\beta_{e-b, e}, \beta_e) + h(\beta_{e, e}, \beta_e))^2]\nonumber\\
&\text{(Triangle inequality and Cauchy-Schwarz)} \nonumber
\\
\ge& \frac{C}{m}\sum_{e\in[2a, m-2a]} h^2(\beta_{e,e}, \beta_{e+a, e+a}) + h^2(\beta_{e,e}, \beta_{e+b, e+b}),
\label{eqn:rearrange hilinger}
\end{align}
where $C$ is a constant. 
Now we are able to group the terms. By the Euclidean algorithm, there exist integers $p$ and $q$ such that $pa+qb=1$. Let $q$ be such an integer with smallest absolute value. We then have $a/b\le |q|\le a$ and $|p|\le |q|$. Therefore, for any $i \in [4, m/4a - 4]$, we can select up to $|p|+|q|$ terms from the range $[4ai-2a, 4ai+2a]$ to group in the above equation. Without loss of generality, assume $y>0$. First, using the Cauchy-Schwarz inequality, we have
\begin{align*}
\text{CIC}_{\zeta, 2/3}(\text{DIST}) \ge &\frac{C}{4ma}\sum_{i}\biggl(\sum_{e\in[4ai-2a, 4ai+2a]}\\ &h(\beta_{e,e}, \beta_{e+a, e+a}) + h(\beta_{e,e}, \beta_{e-a, e-a})\\ 
& + h(\beta_{e,e}, \beta_{e+b, e+b} ) + h(\beta_{e,e}, \beta_{e-b, e-b})\biggr)^2.
\end{align*}
Next, group the terms using the triangle inequality such that whenever combining an $a$-term $h(\beta_{e,e}, \beta_{e^\prime+a, e^\prime + a})$, follow this by combining $\lfloor q/p\rfloor$ $b$-terms $h(\beta_{e,e}, \beta_{e^{\prime\prime}-b, e^{\prime\prime}-b})$. Therefore, the combined term $e^\prime$ is always guaranteed to be in $[4ai-2a, 4ai+2a]$. There might be a concern that a term is combined twice, namely, that there exist $|q_1| \le |q_2|\le|q| $, $|p_1|\le |p_2| \le |p|$ and $(q_1, p_1)\neq (q_2, p_2)$ such that  $p_1 a + q_1 b = p_2 a + q_2 b$. But this is impossible since gcd$(a,b)=1$ and $(p-(p_2-p_1))a+(q-(q_2-q_1))b=1$ contradicts that $q$ has the smallest absolute value. Therefore, we are left with,
\begin{align}
\text{CIC}_{\zeta, 2/3}(\text{DIST}) \ge &\frac{C^\prime}{ma}\sum_{\substack{i\in[4, m/4a-4]\\ e=4ai}}h^2(\beta_{e, e}, \beta_{e+1, e+1}).
\end{align}
Now invoke the Pythagorean lemma of \cite{bar2002information}, stating that $h^2(\beta_{e,e},\beta_{e+1,e+1})\ge1/2(h^2(\beta_{e,e}, \beta_{e+1,e}) + h^2(\beta_{e+1,e}, \beta_{e+1,e+1}))$, as well as the correctness of the protocol, stating that the $h^2(\beta_{e, e}, \beta_{e, e+1})$ terms can each be lower bounded by a positive constant. Thus,
\begin{align}
\text{CIC}_{\zeta, 2/3}(\text{DIST}) \ge &\frac{C^{\prime\prime}}{a^2}.
\end{align}
\end{proof}

The above proof also makes use of the following lemma.
\begin{lemma}
Let $a,b$ be two integers such that $\gcd(a,b) = 1$ and $b<a$. Then there exist integers $x, y$ such that $ax+by =1$. Let $y$ be such an integer with smallest absolute value. Then $b/a\le|y|\le a$ and $|x|\le|y|$.
\end{lemma}
\begin{proof}
To see that $|y|\le a$, w.l.o.g., we assume $y>a$. Then $y=qa+r$, where $r<a<y$ and $q<y$. Thus $ax+b(qa+r)=1$, and $(qb+x) a + rb = 1$ contradicts that $y$ has the smallest absolute value. It is clear that $|y|\ge a/b$. It remains to show that $|x|\le|y|$. This holds since $x=(1-by)/a$ and $b<a$.
\end{proof}

Now we look at a more general case.
\begin{theorem}
\label{thm:theorem pq}
Let $a,b,c$ be positive integers such that $c\neq a$ and $c\neq b$. Suppose there exist integers $p, q$ such that $ap+bq =c$. Let $q$ be such an integer with smallest absolute value. Then any randomized algorithm solving $(a,b,c)$-DIST with probability at least $2/3$ requires $\Omega(n/q^2)$ bits of memory.
\end{theorem}
\begin{proof}
The proof of this theorem is a modification of the proof of the previous proposition. With a similar distribution (replacing $1$ with $c$), we have that Equation (\ref{eqn:rearrange hilinger}) holds. The remaining task is to show a partition scheme of the terms such that $\Theta(m/(|p|+|q|)^2)$ terms survive. Choose an arbitrary $e\in [4a^2, m-4a^2]$. W.l.o.g. assume $p>0$. We can combine the following $|p|+|q|$ terms,
\begin{align*}
h^2(&\beta_{e,e}, \beta_{e+a,e+a}) + h^2(\beta_{e+a,e+a}, \beta_{e+2a,e+2a})\\
& + \ldots +\\
&+ h^2(\beta_{e+pa-a,e+pa-a}, \beta_{e+pa,e+pa})\\ 
&+ h^2(\beta_{e+pa,e+pa}, \beta_{e+pa-b,e+pa-b}) \\
&+ h^2(\beta_{e+pa-b,e+pa-b}, \beta_{e+pa-2b,e+pa-2b})\\
&+\ldots +\\
&+ h^2(\beta_{e+pa+qb+b,e+pa+qb+b}, \beta_{e+pa+qb,e+pa+qb})\\
&\ge\frac{1}{|p|+|q|}[\sum_{i=1}^ph(\beta_{e+(i-1)a,e+(i-1)a}, \beta_{e+ia,e+ia})\\
& + \sum_{j=1}^{|q|}h(\beta_{e+pa-(j-1)b,e+pa-(j-1)b}, \beta_{e+pa-jb,e+pa-jb})]^2\\
&~~~~~~~~~~~(\text{Cauchy-Schwarz})\\
&\ge \frac{h^2(\beta_{e, e}, \beta_{e+1, e+1}) }{|p|+|q|}
~~~~~~~~~~\text{(Triangle inequality)}.
\end{align*}
Choose an $e$. This forbids the choice of another $e$ from the above terms. In fact, $e$ uniquely determines two arithmetic progressions: $e, e+a, e+2a, \ldots, e+pa$ and $e+pa-b, e+pa-2b\ldots, e+pa+qb$. Now, in order to choose a new $e^\prime$, the value $e^\prime$ cannot be a number from the above two progressions, and the progressions determined by $e^\prime$ cannot share a term with one of $e$. Therefore, because we chose $e$, $2(|p|+|q|)$ terms are not available for a new choice of $e^\prime$. By choosing $e$ from $[4a^2, m-4a^2]$ one at a time, we can find $(m-8a^2)/a$ different $e$ for the purpose of combining terms. Denote the set of these $e$ by $S$. In summary, we have,
\begin{align}
\text{CIC}_{\zeta,2/3}(\text{DIST})\ge \frac{C}{m(|p|+|q|)}\sum_{e\in S} h^2(\beta_{e,e}, \beta_{e+1, e+1}).
\end{align}
Since $h^2(\beta_{e,e}, \beta_{e+1, e+1})$ is bounded below by a constant as indicated in the previous proposition, and since $|S| = m-8a^2 = \Omega(m)$, we have $\text{CIC}_{\zeta,2/3}(\text{DIST})=\Omega(1/q^2)$.
\end{proof}

The above lower bound is tight in the sense that there is a protocol that uses $O(n/q^2)$ memory bits and solves $(a, b, c)$-DIST.
\begin{proposition}
\label{thm:theorem pq1}
Let $a, b, c, p, q$ be integers satisfying the same conditions as in the previous theorem. Then there is a randomized algorithm solving $(a,b,c)$-DIST with probability at least $2/3$ using $O(n/q^2) \cdot \textrm{poly}(\log n)$ bits of memory.
\end{proposition}
\begin{proof}
Before the beginning of the stream, we partition $[n]$ into $t$ contiguous pieces each of size $n/t$, where $t= \tilde{O}(n/q^2)$, and the $\tilde{O}$ notation hides logarithmic factors in $n$. For the $i$-th piece of the universe, let the corresponding substream restricted to elements in the $i$-th piece be denoted by $S_i$. For each $S_i$, we maintain a counter $C_i$ for which
\[
C_i = \sum_{l:h[l] = i}v_l\xi_l =\sum_{x\in\{a, b, 1\}}z_{i,x} x
\]
where $\xi_l\sim\{-1, +1\}$ are uniform $4$-wise independent random bits and $z_{i,x} =  \sum_{l:h[l] = i, |v_l|=x}\xi_l$. Let $y_{i,x}$ denote the number of occurrence of $|x|$ in the $i$-th stream. We have $y_{i,x}\le n/t$. By standard concentration arguments, we have that with arbitrarily large constant probability, $|z_{i,x}| = \tilde{O}(\sqrt{y_{i,x}}) = \tilde{O}(\sqrt{n/t})$. We can select an appropriate $t = \tilde{O}(n/q^2)$ for which this implies $|z_{i,x}|$ is at most $|q|/4$. Now the claim is that by reading the value $C_i\mod a$, this is enough to distinguish whether there is a frequency of absolute value ``$c$'' in the stream or not. This follows since the sets of values of $C_i\mod a$ in the two cases are disjoint. To see why, note that if $z_b^\prime b = z_bb + c\mod a$, then we have $(z_b^\prime - z_b) - ra = c$. However, since $|z_b^\prime - z_b|<|q|$, by the minimality of $|q|$, this is a contradiction. 
\end{proof}

\subsection{{\sf ShortLinearCombination} Problem for Multiple of Frequencies}
\begin{definition}
Let $u=(u_1, u_2, \ldots u_r)$ be a vector in $\mathbb{Z}^r$ for an integer $r$.
Let $d>0$ be an integer not in the vector $u$. A stream $S$ with frequency vector $v$ is given to an algorithm, where $v$ is promised to be from $V_0=\{u_1, u_2, \ldots, u_r, 0\}^n$ or $V_1=\{\text{EMB}(v, i, e)\mid v\in V_0, i\in[n], e\in\{-d, d\}\}$. The
$(u,d)$-DIST problem is to distinguish whether $v\in V_0$ or $v\in V_1$. 
\end{definition}
\begin{theorem}
Let $u=(a_1, a_2, \ldots, a_r)$ be a set of integers and $d>0$ be a positive integer such that $d=q_1a_1+q_2a_2+\ldots q_r a_r$ where $q_1, q_2, \ldots q_r$ are integers. These integers are choosen such that $q=\sum_{i}|q_i|$ is minimal. Then any randomized algorithm solving $(u, d)$-DIST with probability at least $2/3$ requires $\Omega(n/q^2)$ bits of memory.
\end{theorem}
\begin{proof}
The proof of the lower bound is a straightforward extension of the previous argument for three frequencies. Let $a=\max(a_1, a_2, \ldots, a_r)$. The number of combined terms is $m/(|q_1|+|q_2|+\ldots |q_r|)$. Therefore, the lower bound is $\Omega({n}/q^2)$. The upper bound is a also a straightforward extension of the three frequency case.
\end{proof}

\section{Nearly Periodic Functions }\label{app:nearlyperiodic}
To summarize some of the discussion so far, when a function has a large decrease in value we can often get a large lower storage bound using an index reduction.
However, the index reduction fails on some functions for which we have no sub-polynomial space algorithm and no other bound is known, these are the nearly periodic functions.
This section describes a few, disconnected, properties of nearly periodic functions.
The purpose is to understand as much as possible about the nearly periodic functions.

The next subsection presents an example of a $1$-pass tractable nearly periodic functions, this function has a very nice structure such that it is one pass tractable by making use of modulo.  Section \ref{sec: metrizing G} defines a natural metric on $\G$ and shows that every $\bbS$-nearly periodic function, which may or may not be 1-pass tractable, is arbitrarily close to a $1$-pass intractable function while every 2-pass tractable function is far from every 2-pass intractable function.  
Section~\ref{sec: nearly periodic transform} defines a transformation with a similar property to the metric of Section~\ref{sec: metrizing G} in that it separates $\bbS$-nearly periodic functions from 1-pass tractable $\bbS$-normal functions.
Section~\ref{sec: counting nearly periodic} establishes a discretized model for the $g$-SUM problem and considers counting the analogue of the tractable and nearly periodic functions.
The main result is that there are many fewer nearly periodic functions because of the strong constraints on their structure.
Finally, Section~\ref{sec: nearly periodic asymptotic} establishes a strange necessary condition on nearly periodic functions that are bounded above by a sub-polynomial function.

\subsection{A $1$-pass Tractable Example\label{example:nearly periodic}}
Constructing a tractable nearly periodic function turns out to be a non-trivial exercise.
This section provides such a construction.

\begin{definition}
Let $x =\sum_{j=0}^\infty a_j2^j\in \mathbb{N}$, for $a_j\in\{0,1\}$, be the binary expansion of the integer $x$.  Define $i_x:=\min\{j: a_j = 1\}$, the location of the first non-zero bit of $x$.
The function is $g_{np}(x) = 2^{-i_x}$, for $x>0$, and $g_{np}(0)=0$.
\end{definition}

\begin{proposition}
$g_{np}$ is $\mathbb{S}$-nearly periodic.
\end{proposition}
\begin{proof}
Since there is an infinite sequence $\{1, 2, 4, \ldots, 2^n\ldots\}$ such that $g_{np}(x)=1/x$, the first condition of near-periodicity is satisfied. It is remains to show that the second condition is also satisfied.

Let $\gamma> 0$ be a constant. Consider any integer $x>0$. By construction $g_{np}(x) = 2^{-i_x}$. Let $y>x$ be another integer. Similarly, we have $g_{np}(y) = 2^{-i_y}$.
If $g_{np}(x)/g_{np}(y) = 2^{i_y-i_x} \ge y^\gamma$, then $i_y - i_x \ge \gamma \lg y$. Choose $N=\lceil 2^{1/\gamma}\rceil$. For any $y>N$, we have $i_y-i_x > 1$ and therefore $i_{x+y}=i_x$. Thus, $g_{np}(x+y) = g_{np}(x)$.
\end{proof}

\begin{proposition}
$g_{np}$ is $1$-pass tractable.
\end{proposition}
\begin{proof}

We demonstrate that one can find $(g_{np}, \lambda)$- heavy hitters in $\poly(\lambda^{-1}\log{n}\log{M})$ space.
It is sufficient to demonstrate an algorithm that will find a single $(g_{np},\lambda)$-heavy hitter, since we can reduce the heavy hitters problem to this case by hashing the stream $O(\lambda^{-2})$ ways and running this algorithm on each substream.

Suppose that $j^*$ is the single $(g_{np},\lambda)$-heavy hitter with frequency $x$.
Let $v_1,v_2,\ldots,v_{n}\geq 0$ be the frequencies in the stream, and let $U=\{j: i_{v_j}\le i_x\}$.
By definition $g_{np}(v_j)\geq g_{np}(x)$ for $j\in U$, hence $|U| \le 1+\lambda^{-1}\leq 2\lambda^{-1}$. 

Let $C=O(\lambda^{-2})$.
Apply a uniform hash function $h:[n]\rightarrow [C]$ to separate the stream into $C$ substreams $S_1,S_2,\ldots,S_{C}$.
With constant probability, no two elements of $U$ are in the same substream, so suppose that this happens. 

On each substream $S_k$ with frequencies $v^{(k)}_1,\ldots,v^{(k)}_n$ we run the following algorithm $D=O(\log{n})$ times independently and in parallel:
\begin{itemize}
\item Sample pairwise independent $X_1,\ldots,X_n\sim\text{Bernoulli}(1/2)$ random variables. 
\item Compute $m = \sum_{j}X_jv^{(k)}_j$ and $i_m$. 
\item Output $2^{-i_m}$.
\end{itemize}

Consider the set of $D$ values output by this algorithm. 
If there is a single item $j^*$ with minimum $i^*:=i_{v^{(k)}_{j^*}}$, then a Chernoff bound implies that very nearly $D/2$ values are equal to $2^{-i^*}$, and this is the maximum value among all of the pairs.  
In this case, the label $j^*$ can be found in post-processing by binary search. 
Let $X_{j,\ell}$ denote the value of the $j$th Bernoulli variable on the $\ell$th trial.
Specifically, one finds the set $M\subseteq[D]$ of trials for which $2^{-i_m}$ is equal to the maximum among the $D$ values, and with high probability, only $j^*$ will satisfy $X_{j,\ell}=1$ for all $\ell\in M$ and $X_{j,\ell}=0$ for all $\ell\in[D]\setminus M$. 
We can detect the case where there is more than one item with minimum $i_{v^{(k)}_j}$ because either the number of maximizing values will be too large or the binary search will fail to yield a unique element.

The single-heavy-hitter algorithm now outputs the pair $(j^*,2^{-i^*})$ if the number of maximizing values is correct and the binary search yields a unique element or nothing if either of those conditions fails.
With the extra hashing step mentioned at the beginning this yields a 1-pass $O(\lambda^{-4}\log{n}\log{M})$-space $(g_{np},\lambda)$-heavy hitters algorithm, thus $g_{np}$ is 1-pass tractable.
\end{proof}

\subsection{Asymptotic Behavior of Nearly Periodic Functions\label{sec: nearly periodic asymptotic}}
This section proves a necessary function for a bounded function to be nearly periodic.  

\strangeLimits*
\begin{proof}
Since $g$ is positive and nearly periodic there exists $\alpha>0$, a strictly increasing sequence $y_k$, and a sequence $x_k$ such that $g(x_k)\geq y_k^{2\alpha} g(y_k)$ for all $k$.
Since $g$ is bounded above by a sub-polynomial function this implies that $g(y_k)\leq y_k^{-\alpha}$ for all sufficiently large $k$, which proves the first claim.
Letting $\epsilon_k = (\log{y_k}\max\{g(x)\mid x\leq y_k\})^{-1}$, near periodicity of $g$ implies $|g(x)-g(x+y_k)|\leq \epsilon_k g(x)$ for sufficiently large $k$.  
Finally, $\epsilon_k g(x)\to 0$ since $g$ by definition, which completes the proof.
\end{proof}

It would not be surprising if for every $x$ there exists a sequence $y_k$ such that $g(x+y_k)\to g(x)$.
The novelty in Proposition~\ref{prop: strange limits} is that there exists a single sequence $y_k$ such that the convergence holds at every point $x$. 

Let us draw a final comparison with periodic functions.
Take, for example, a function $g$ that is periodic with period $p=\min\{x\in\N: g(x)=0\}$. 
$g$ is bounded above by a constant, hence any sequence $y_k$ of $\alpha$-periods has $g(y_k)\to 0$.
This implies that $g(y_k)=0$ for all sufficiently large $k$, and hence $y_k$ is a multiple of the period $p$ (for all sufficiently large $k$). 
Of course, $g(x+y_k)=g(x)$ when $y_k$ is a multiple of $p$, by periodicity, hence $g(x+y_k)\to g(x)$ trivially.

\subsection{A Slow-Varying Transformation\label{sec: nearly periodic transform}}
The predictability is preserved under the following transformation.

\begin{definition}
\label{def:l_eta_transformation}
Given a function $g$, we define an {\em $L_\eta$-transformation} of $g$ as $L_\eta(g)(x):=g(x) \log^{\eta}(1+x)$ for $\eta>0$.
\end{definition}

\smallChangeOPass*
\begin{proof}
Since $g$ is $\bbS$-normal and $1$-pass tractable, then $g$ is slow-dropping, slow-jumping and predictable.
It is suffice to show that $L_\eta(g)$ is also predictable. 

By definition, if $g$ is predictable, for any sub-polynomial $\epsilon^\prime$ and $0<\gamma<1$, there exists $N$ such that if $x>N$, for all $y\in[1,x^{1-\gamma})$, if $(x+y)\notin\delta_{\epsilon^\prime}(g,x)$ then $g(y)\ge g(x)x^{-\gamma}$.

Let $g^\prime = L_\eta(g)$ for $\eta >0$ . Consider a sub-polynomial $\epsilon$ and $0<\gamma<1$, for $y\in[1, x^{1-\gamma})$, such that $x+y\not\in\delta_\epsilon(g^\prime, x)$, we have: if $g^\prime(x+y)\ge g^\prime(x)$ then 
\begin{align*}
\left|\frac{g(x+y)\log^\eta x(1+x^{-\gamma})}{g(x)\log^\eta x}-1\right|
&\ge 
\left|\frac{g(x+y)\log^\eta(x+y)}{g(x)\log^\eta(x)}-1\right|\\
&>\epsilon.
\end{align*}
Hence 
\[
\frac{g(x+y)}{g(x)}\left(1+\frac{\log (1+x^{-\gamma})}{\log x}\right)^\eta>1+\epsilon.
\]
Then $g(x+y)/g(y)\ge 1+\epsilon-\eta x^{-\gamma}/\log x\ge 1+\epsilon/2$
 for sufficiently large $x$, we have $g(y)\ge x^{-\gamma}g(x)$.

If $g^\prime(x+y)<g^\prime(x)$, then we have 
\[
\left|\frac{g(x+y)}{g(x)}-1\right|\ge \left|\frac{g(x+y)}{g(x)}\frac{\log^\eta x(1+x^{-\gamma})}{\log^\eta x}-1\right|>\epsilon
\]
we still have $g(y)\ge x^{-\gamma}g(x)$ for sufficiently large $x$.
\end{proof}

\smallChangePeriodic*
\begin{proof}
Suppose $g$ is $1$-pass tractable. Let $l(x) = \log(x)^\eta$.
By Lemma~\ref{lem:polynomialperiods}, we can find a sequence $\{(x_i, y_i)\}_{i=1}^\infty$ such that $x_i^{1+\gamma}<y_i$, $x_i$  strictly increasing,  $g(x_i)\ge  g(y_i)y_i^{\alpha}$ for some $\alpha >0$ and $\gamma >0$. For any $0<\epsilon<\gamma\eta/[2(1+\gamma\eta)]$, $|g(x_i + y_i)-g(x_i)|\le \epsilon g(x_i)$ holds after finite terms.  Since $x_i+y_i\ge x_i^{1+\gamma}$ , $|l(x_i+y_i)-l(x_i)|\ge \gamma \eta l(x_i)$.
Note that $(1+c)(1+\gamma\eta)>1+\gamma\eta/2$ holds for any $c\in[-\epsilon, \epsilon]$.
Therefore,
\begin{align}
\frac{|l(x_i+y_i)g(x_i+y_i)-l(x_i)g(x_i)|}{l(x_i)g(x_i)}\ge \gamma\eta/2
\end{align}
for sufficiently large $i$. Thus, $l(x) g(x)$ is not nearly periodic. Also due to $g(x_i)\ge g(y_i)y_i^\alpha$, this function is not slow-dropping.
\end{proof}

\begin{lemma}
\label{lem:polynomialperiods}
If $g\in\G$ is nearly periodic, then $g$ is either $1$-pass intractble or there exists $\alpha>0, \gamma >0$ and a sequence $\{(x_i, y_i)\}_{i=1}^\infty$ such that $x_i^{1+\gamma}<y_i$, $x_i$ strictly increasing,  $g(x_i)\ge  g(y_i)y_i^{\alpha}$ and for any sub-polynomial $\epsilon$, $|g(x_i + y_i)-g(x_i)|\le \epsilon g(x_i)$ after some finite terms.
\end{lemma}
\begin{proof}
Suppose $g$ is $1$-pass tractable. 

Since $g$ is nearly periodic, there exists $\alpha^*>0$, and an infinite sequence $\{(x_i, y_i)\}_{i=1}^\infty$ such that $x_i<y_i$, $y_i$ strictly increasing, $g(x_i)\ge  g(y_i)y_i^{\alpha^*}$ and,
for any \subpoly{}~$\epsilon$,   w.l.o.g.,  $|g(x_i + y_i)-g(x_i)|\le \epsilon(y_i) g(x_i)$ after some finite terms. For any $0<\beta<\alpha^*$, $\gamma < \beta / 2$, the claim is that, there exists a subsequence $T=\{x_{i_j},y_{i_j}\}_{j=0}^\infty$ such that \[ g(y_{i_j})\ge y_{i_j}^{-\beta}g(2^{\lceil1+\gamma\rceil}y_{i_j}y_{i_j}^{\gamma}).\] Otherwise, if $y_{i}^{-\beta}g(2^{\lceil1+\gamma\rceil}y_{i}y_{i}^{\gamma}) > g(y_{i})$ for some infinite $i$s, we can use a standard reduction from DISJ$(2^{\lceil1+\gamma\rceil}y_{i}^{\gamma},y_{i}^{\beta})$, which contradicts that $g$ is $1$-pass tractable.
Therefore, for the sub-sequence $T$, let $\alpha^\prime=(\alpha^*-\beta)/(1+\gamma)$, $g(x_{i_j})\ge (y_{i_j}^{1+\gamma})^{\alpha^\prime}g(2^{\lceil1+\gamma\rceil }y_{i_j}^{1+\gamma})$ for any $j$. 
Let $x_{i_j}^\prime:=x_{i_j}+y_{i_j}< 2 y_{i_j}$ and $y_{i_j}^\prime:=2^{\lceil1+\gamma\rceil}y_{i_j}^{1+\gamma}$. We have $g(x_{i_j}^\prime)\ge y_{i_j}^{\prime\alpha^\prime(1-0.1)}g(y_{i_j}^{\prime})$, ${x_{i_j}^\prime}^{1+\gamma} < y_{i_j}^{\prime}$ for sufficiently large $j$. The property of strict increasing can also be guaranteed by choosing the proper $j$s. Let $\alpha = \alpha^\prime(1-0.1)$. $T^\prime =\{x_{i_j}^\prime, y_{i_j}^\prime\}_{j=0}^\infty$ is the desired sequence since by nearly periodicity,  for any \subpoly{}~$\epsilon$, $|g(x_{i_j}^\prime + y_{i_j}^\prime)-g(x_{i_j}^\prime)|\le \epsilon(y_{i_j}^\prime) g(x_{i_j}^\prime)$ must hold after finite terms.
 
\end{proof}

\subsection{Counting Nearly Periodic Functions in a Discretized Setting}\label{sec: counting nearly periodic}

Let us consider functions taking values between 0 and 1 on the interval $[M]_0=\{0,1,2,\ldots,M\}$ with fixed-point precision $1/M'$.
$M,M'\in\N$ are parameters that we can choose later (we will choose them to be $\poly(n)$).
Upon rescaling, such a function can be thought of as a function $g:[M]_0\to[M']_0$.
As before, we will request that $g(0)=0$ and $g(x)>0$, for $x>0$.
Our previous scaling choice (i.e.\ $g(1)=1$) becomes $g(1)=M'$.
Let $\G_D$ be the set of all functions $g:[M]_0\to[M']_0$ with $g(0)=0$, $g(1)=M'$, and $g(x)>0$ for $x>0$.
Henceforth, we refer to this as the ``discretized model''.

The plan for this section is to translate the tractable and nearly periodic functions into the discretized model and show that there are many fewer nearly periodic functions.
The main result is Theorem~\ref{thm: few nearly periodic functions}, where $B_n$ is the discretized analogue of the $\bbS$-nearly periodic functions and $T_n$ is the discretized analogue of the 1-pass tractable functions.
\begin{theorem}\label{thm: few nearly periodic functions}
If $M,M'=\poly(n)$ then
\[\frac{|B_n|}{|T_n|}\leq 2^{-\Omega(M\log\log{n})}\leq2^{-n^{\Omega(1)}}.\]
\end{theorem}

The discretized model is a bit different than our previous one.
Perhaps the biggest difference is that here we choose the size of the problem (i.e.\ the values $n$, $M$, and $M'$) at the same time as we choose the function.
This is done in order to make it possible to count the functions.
\begin{remark}
Some difference is necessary as the definitions, from Section~\ref{sec: preliminaries}, of these functions classes (1) allow values from the continuum and (2) are defined by their asymptotics.  
Especially (2) is relevant. Suppose we, instead, take the set of tractable functions bounded by 1, round the values up to the nearest multiple of $1/M'$ and restrict the domain to $[M]_0$.
The result, a consequence of the asymptotic nature of the definition of tractable functions, is that we find every function in~$\G_D$!
Of course, the same is true for the nearly periodic functions.
Hence, it proves impossible to differentiate normal and nearly peroidic functions by rounding and restricting the functions.
\end{remark}

Next we will translate the classes of tractable functions and nearly periodic functions into this discretized setting and bound their sizes.
Clearly $|\G_D|=(M')^{M-1}$.
Let $T_n\subseteq \G_D$ be the set of functions for which there exists a turnstile $(1\pm1/2)$-approximation algorithm for $g(f)$ using $O(\log^3{n}\log{M})$ bits of storage.

\begin{lemma}\label{lem: Tn is big}
\[|T_n|\geq \left(M'-\frac{M'}{\log{n}}\right)^{M-1}\]
\end{lemma}
\begin{proof}
  We can approximate in $O(\log^3{n}\log{M})$ bits any function with minimum value at least $\lfloor M'/\log{n}\rfloor$.
  The number of such functions is at least \[\left(M'-\frac{M'}{\log{n}}\right)^{M-1}.\]
\end{proof}

The raison d'\^etre for the class of nearly periodic functions are the following two properties: (1)~the functions sustain a large drop in value, infinitely many pairs $x<y$ satisfy $g(x)\geq y^{\alpha}g(y)$, but (2) the standard index reduction fails to give a lower bound because the function approximately repeats itself after the drop, $|g(x)-g(x+y)|<g(x)/h(y)$ for any sub-polynomial function $h$ and sufficiently large $y$.
The lower bound of Proposition~\ref{prop: lower bound} motivates our definition for the nearly periodic functions in the discretized model.
Specifically, nearly periodic functions in the discrete model will satisfy the first hypothesis, $g(x)\geq sg(y)$, for large $s$ but will not satisfy either of items \ref{it: minus constraint} or \ref{it: add constraint}.

We will allow a turnstile component to the reduction, which makes the proof a little easier.
One can make a different translation of the nearly periodic functions into this discretized model that avoids the turnstile model, but we feel that this is the simplest.

\begin{proposition}\label{prop: lower bound}
Let $x,y\in[M]$ such that $g(x)\geq sg(y)$, for some $1\leq s\leq n$. 
If at least one of the following holds:
\begin{enumerate}
  \item\label{it: minus constraint} $|g(x) - g(|y-x|)|\geq g(x)/\log^2{n}$,
  \item\label{it: add constraint} $x+y\leq M$ and $|g(x) - g(y+x)|\geq g(x)/\log^2{n}$
\end{enumerate}
then any algorithm that gets a $(1\pm1/2)$-approximation to $g(f)$ uses $\Omega(s/\log^2{n})$ bits of storage in the worst case.
\end{proposition}
\begin{proof}
  This is proved with a simple modification of the INDEX reduction used in Theorem~\ref{thm:normintractable}.
  For case \ref{it: minus constraint}, Alice uses the frequency $y$ for each of her items and Bob use the frequency $-x$.
  For case \ref{it: add constraint}, Alice uses the frequency $y$ for each of her items and Bob uses $x$.
\end{proof}

Let $B_n\subseteq\G_D$ be the set of functions $g$ for which (1) there exists $x,y\in[M]$ such that $g(x)\geq (\log{n})^8 g(y)$ and (2) for all $x,y\in[M]$ satisfying $g(x)\geq \frac{1}{2}(\log{n})^8g(y)$ we have $|g(x)-g(|y-x|)|<g(x)/\log^2{n}$ and also $|g(x+y)-g(x)|<g(x)/\log^2{n}$ provided $x+y\leq M$.

Condition (1) corresponds to the first condition in the definition of the nearly periodic functions and condition (2) to the second, hence $B_n$ contains the analogue of the nearly periodic functions in this discretized model.
Albeit, in both cases we only request a polylogarithmic gap, rather than a polynomial one, which means $|B_n|$ is an over-estimate for the number of nearly periodic functions in the discretized model.

We will derive an upper bound on $|B_n|$ as follows.
First, we chose the location of the least minimizer $y\in[M]$ of the function and its value $g(y)\leq M'/\log^8{n}$--trivially, there are fewer than $MM'$ possibilities.
The least maximizer is $x=1$ with value $g(x)=M'$.
Now, given any point $i\in[M]$, no matter how we choose its value, we will have either $g(i)\geq g(y)\log^4{n}$ or $g(x)\geq g(i)\log^{4}{n}$, so that after selecting a value for $g(i)$ we can try to apply the constraints described above to the values $g(x+i)$, $g(|x-i|)$, $g(y+i)$, and $g(|y-i|)$, thus limiting the number of choices significantly.

The next step is to select the set $U=\{i\in[M]:g(i)\geq \sqrt{g(x)g(y)}\}$ from the $2^{M-2}$ available choices.
Let $L=[M]\setminus U$.
Notice that for any $i\in U$ we have $g(i)\geq g(y)\sqrt{\frac{g(x)}{g(y)}}\geq g(y)\log^4{n}$, and similarly for $i\in L$ we have $g(i)\log^4{n}< g(x)$.

The following lemma is used in the proof of our bound on $|B_n|$.
\begin{lemma}\label{lem: big index subset}
Given $S\subseteq [M]$ and $j\in[M]$, there exists a set $W\subseteq S\times[M]$ of pairs $(i,|i-j|)$ with size at least $\frac{1}{4}|S|-1$ such that all of the values in the pairs are distinct. 
\end{lemma}
\begin{proof}
Let $G$ be the directed graph with vertices $[M]$ and an edge $(i,|i-j|)$, for all $i\in S\setminus\{j,j/2\}$.
Any matching $W\subseteq E(G)\subseteq S\times[M]$ is a set of pairs with all of the values distinct.
It remains to show that there is a matching $W$ with size $|W|\geq \frac{1}{4}|S|-1$.
The out-degree of each vertex $k$ in $G$ is one if $k\in S$ and zero otherwise, and the in-degree of $k$ is zero if $k\geq j$ and at most two for $k<j$.
Furthermore, each cycle in $G$ has only two edges.
For each vertex of in-degree two delete one of the incident edges as follows: if the vertex is on a cycle delete the cyclic edge $(u,v)$ with $u<v$, otherwise delete an arbitrary edge.
This removes at most $|E|/2$ edges from the graph.
The remaining graph has at least $|E|/2$ edges and consists entirely of paths and isolated 2-cycles, thus it contains a matching of size at least $\frac{1}{4}|E|\geq \frac{1}{4}(|S|-2)$.
\end{proof}

\begin{lemma}\label{lem: Bn is small}
\[|B_n| \leq \frac{4^M M (M')^{M+1}}{(\log{n})^{\frac{1}{8}M-1}}\] 
\end{lemma}
\begin{proof}
  Let $x,y\in[M]$ be the least maximizer and minimizer, respectively, of the function.
  Then $x=1$ with $g(x)=M'$ and there are trivially no more than $MM'$ choices for the values of $y$ and $g(y)$, which must satisfy $g(y)\log^8{n}\leq g(x)$.
  There are $2^{M-2}$ choices for $U=\{i\in[M]: g(i)\geq \sqrt{g(x)g(y)}\}$.

  Suppose that $|U|\geq M/2$; the case $M-|U|=|L|\geq M/2$ is handled similarly. 
  Apply Lemma~\ref{lem: big index subset} with $S=U$ and $j=y$ to find the set $W$ of size at least $\frac{1}{8}M-1$.
  Next assign the values $g(i)$ for every $i\in [M]$ that is the first coordinate of a pair in $W$ or is not in any pair in $W$.
  With the trivial bound, there are no more than $(M')^{M-|W|}$ possibilities.
  For the remaining points $|y-i|$ we must have that
  \[|g(i)-g(|y-i|)|\leq \frac{g(i)}{\log^2{n}}\leq \frac{M'}{\log^2{n}},\]
  so there are no more than $2M'/\log^2{n}$ choices for each of these $|W|$ remaining points.
  Hence, the total number of functions in $B_n$ is no more than
  \[(MM')2^M M'^{M-|W|}\left(\frac{2M'}{\log^2{n}}\right)^{|W|}\]
  which completes the proof.
\end{proof}

\subsection{Metrizing $\G$\label{sec: metrizing G}}

Given two functions $g,h\in\G$ let 
\begin{align*}
\Theta(g,h) &:= \sup_{x\in\N}\log\left(\max\left\{\frac{g(x)}{h(x)},\frac{h(x)}{g(x)}\right\}\right)\\
&= \sup_{x\in\N}|\log g(x)-\log h(x)|\\ 
&= \|\log{g}-\log{h}\|_\infty.
\end{align*}

$\Theta$ is an extended metric\footnote{An extended metric satisfies all of the properties of a metric except that it is allowed to take the value $\infty$.} on $\G$.

The purpose of this section is to prove that $\bbS$-nearly periodic functions are unstable, in a sense described by $\Theta$, while 2-pass tractable $\bbS$-normal functions are stable.
More precisely, we can think about a function $g$ and a perturbed version $h$ of $g$, where $h$ is found by making a small, relative change to $g$ at each point.
In this case $\Theta(g,h)$ is small.
If $g$ is $\bbS$-normal and 2-pass tractable, then Proposition~\ref{prop:metric slow jumping dropping} shows that $h$ is also 2-pass tractable (even if the perturbation is very large).

\begin{proposition}
\label{prop:metric slow jumping dropping}
Let $g,h\in\G$ with $\Theta(g,h)<\infty$.
If $g$ is slow-dropping and slow-jumping then so is $h$.
\end{proposition}
\begin{proof}
Let $x,y\in\N$. 
Then
\[\log{\frac{h(x)}{h(y)}}\leq \log\left(\frac{h(x)g(y)}{g(x)h(y)}\cdot \frac{g(x)}{g(y)}\right)\leq \Theta(g,h)^2\log\frac{g(x)}{g(y)}.\]
Thus, $h(x)/h(y) \leq e^{\Theta(g,h)^2}g(x)/g(y)$ and the result follows from the definitions of slow-dropping and slow-jumping.
\end{proof}

To the contrary, if $g$ is $\bbS$-nearly periodic, then the next theorem shows that one can always find a nearby function $h$ that is $1$-pass intractable.

\begin{theorem}
If $g\in\G$ is $\bbS$-nearly periodic then for every $\delta>0$ there exists a normal $1$-pass intractable function $h\in G$ such that $\Theta(g,h)\leq\delta$.
\end{theorem}
\begin{proof}
Since $g$ is $\bbS$-nearly periodic, there exists $\alpha>0$ and a sequence of pairs $(x_k,y_k)$ such that $y_k$ is increasing, $x_k<y_k$, and $g(x_k)\geq y_k^\alpha g(y_k)$.
By taking a subsequence we can ensure that all of the values $x_k$, $y_k$, and $x_k+y_k$ are distinct.
Define $h$ as follows: $h(x)=g(x)$ for $x\notin\cup_j\{x_j,x_j+y_j\}$, $h(x_k)=(1+\delta)g(x_k)$, and $h(x_k+y_k)=g(y_k)/(1+\delta)$.
Then $\Theta(g,h)=\log(1+\delta)\leq\delta$.

By construction, $h(x_k)/h(y_k)\geq g(x_k)/g(y_k)\geq y_k^\alpha$, but $|g(x_k)-g(x_k+y_k)|\geq (1+\delta)g(x_k)$.
Thus $h$ is not slow-dropping and not $\bbS$-nearly periodic, so $h$ is not 1-pass tractable by Lemma~\ref{lem:normslowdropping}.
\end{proof}

\subsection{Towards Characterizing Nearly Periodic Functions}
The set of nearly-periodic function defeat the standard reduction from the $\disj$ and $\textsf{INDEX}$ problems. 
Theorem~\ref{thm:communication nearly periodic} shows how one can use the \textsc{ShortLinearCombination} problem, of Definition~\ref{def:SLC}, to provide space lower bounds on $g$-SUM algorithms for some nearly periodic functions~$g$.

\begin{definition}
For any $N>0$ and non-increasing sub-polynomial function $h(x)$, define the $(N, h)$-dropping set of $f$ to be
\[
\mathcal{D}_{N,h}(f) = \{x\mid 1\le x\le N, f(x)\le h(N)/N\}.
\]
\end{definition}

\begin{proposition}
Let $f$ be a nearly-periodic function, then for all $n_0>0$, there exists $N>n_0$ and a sub-polynomial function $h$ such that $|D_{N,h}|>0$.
\end{proposition}
\begin{proof}
Choose $N=n_0+1, h(x) = f(1)N$, i.e. $h(x)$ is a constant function. Then $1$ in $D_{N,h}$ because $1 \le N$ and $f(1) = h(N)/N$.
\end{proof}

\begin{definition}
An \emph{$\alpha$-indistinguishable} frequency set of $[n]$ is a tuple $(s, d)$, where $s\subseteq [n]$ is a subset and $d\in[n], d\notin s$ is a integer such that $(s, d)$-DIST problem requires $\Omega(n^\alpha)$ space. 
\end{definition}

\begin{theorem}
\label{thm:communication nearly periodic}
Let function $f:\mathbb{R}\rightarrow \mathbb{R}^+$, symmetric, nearly-periodic and $f(0) = 0$. If there exists a non-increasing sub-polynomial function $h(x)$ and constants $\alpha>0, 0<\delta<1$, such that for any $n_0>0$, there exists $N>n_0$ and a integer set $|S|>0$ satisfying the follows,
\begin{enumerate}
\item $0<|\mathcal{D}_{N,h}(f)|\le N^{1-\delta}$;
\item $S\subseteq |\mathcal{D}_{N,h}(f)|$ and $d\in[N]$ such that $(S,d)$ is a $\alpha$-indistinguishable frequency set of $N$;
\item $f(d)\ge h(d)$,
\end{enumerate}
then for any $n_0>0$, there exists a stream of domain $N>n_0$,  any one-pass randomized streaming algorithm $\mathcal{A}$ that outputs a $1\pm \epsilon(N)$ approximation of $f$-SUM with probability at least $2/3$,  requires space $\Omega(n_0^{\alpha})$ bits, where $\epsilon(x)<1$ is a sub-polynomial function.
\end{theorem}
\begin{proof}
Suppose we have an algorithm $\mathcal{A}$ that gives a $1\pm \epsilon(N)$ approximation to $f$-SUM with probability at least $2/3$. We can now use $\mathcal{A}$ to construct a protocol that solves $(S,d)$-DIST problem on domain $N$. The algorithm is straightforward: run $\mathcal{A}$ on the input stream.  If the result $\mathcal{A}$ outputs is $\le h(d)/N^\delta$, then there is no frequency of absolute value $d$ in the stream. This is so since the $g$-SUM is at most $\sum_{a\in S}f(a)\le N^{1-\delta}/N$. Otherwise, if the output is at least $h(d)(1-\epsilon(N))$. By the guarantee of $\mathcal{A}$,  it can distinguish the 2 cases. The space needed of $\mathcal{A}$ inherits the lower bound of $(S, d)$-DIST problem, thus $\Omega(N^{\alpha})$ bits.
\end{proof}

\end{appendix}
\pagenumbering{gobble}

\end{document}